%% file: arxiv.tex
\documentclass[a4paper,UKenglish,cleveref, autoref, thm-restate, authorcolumns]{lipics-v2021}
\usepackage{amsthm,amssymb,amsmath}  
\usepackage{xspace,enumerate}
\usepackage[noline,noend,ruled]{algorithm2e}

\usepackage{todonotes,tikz} 
\usetikzlibrary{snakes,calc}

\newcommand{\ceil}[1]{\lceil #1 \rceil}
\newcommand{\Triples}{\mathsf{zip}}
\newcommand{\R}{\mathbf{R}}

\newcommand{\floor}[1]{\lfloor #1 \rfloor}

\newcommand{\defproblem}[3]{
  \vspace{2mm}
  \begin{center}
  \noindent\fbox{
  \begin{minipage}{0.9\textwidth}
  \vspace{2mm}
  \textbf{#1}

  \smallskip
  \noindent
  {\bf{Input:}} #2
  
  \smallskip
  \noindent
  {\bf{Output:}} #3
  \vspace{2mm}
  \end{minipage}
  }
  \end{center}
  \vspace{2mm}
}

\newcommand{\q}{\hat{q}}
\def\dd{\mathinner{.\,.}} 
\newcommand{\cO}{\mathcal{O}}
\newcommand{\Oh}{\mathcal{O}}
\newcommand{\ctO}{\mathcal{\tilde{O}}}
\newcommand{\cOtilde}{\mathcal{\tilde{O}}}
\newcommand{\ed}{\delta_E}

\newcommand{\rot}{\mathsf{rot}}

\newcommand{\cycoc}{\mathsf{CircOcc}}

\newcommand{\Ext}{\mathsf{Ext}}
\newcommand{\lcut}{\mbox{\sf{L-cut}}}
\newcommand{\rcut}{\mbox{\sf{R-cut}}}

\newcommand{\LCP}{\mathsf{LCP}}
\newcommand{\D}{\mathsf{D}}
\newcommand{\DD}{\mathbf{NonOv}}
\newcommand{\Frag}{\mathit{Frag}}
\newcommand{\First}{\mathit{first}}
\newcommand{\shift}{\mathit{shift}}
\newcommand{\locked}{\mathit{locked}}
\newcommand{\ZZ}{\mathcal{Z}}

\def\EPSM{{\sc PeriodicSubMatch}\xspace}
\newcommand{\Occ}{\textsf{Occ}}

\newcommand{\scope}{\mathsf{scope}}
\newcommand{\CritPos}{\mathsf{CritPos}}

\newcommand{\anchored}{\mathtt{Anchored}}
\newcommand{\anyanchored}{\mathtt{AnyAnchored}}

\newcommand{\implied}{\mathsf{Implied}}

\newcommand{\Chain}{\mathrm{Chain}}

\newcommand{\edl}[2]{{\delta_E}(#1,{}^*\!#2^*)}
\newcommand{\eds}[2]{{\delta_E}(#1,{}^*\!#2)}
\newcommand{\edp}[2]{{\delta_E}(#1,#2^*)}
\newcommand{\Access}[0]{\mathsf{Access}}
\newcommand{\Extract}[0]{\mathsf{Extract}}
\newcommand{\Length}[0]{\mathsf{Length}}
\newcommand{\IPM}[0]{\mathsf{IPM}}

\newcommand{\rright}{\textsf{right}}
\newcommand{\lleft}{\textsf{left}}

\newtheorem{fact}[theorem]{Fact}

\theoremstyle{plain}

\def\pillar{{\tt PILLAR}\xspace}

\title{Approximate Circular Pattern Matching under~Edit~Distance}

\hideLIPIcs

\author{Panagiotis Charalampopoulos}{Birkbeck, University of London, UK}{p.charalampopoulos@bbk.ac.uk}{https://orcid.org/0000-0002-6024-1557}{}
\author{Solon P. Pissis}{CWI, Amsterdam, The Netherlands \and Vrije Universiteit, Amsterdam, The Netherlands}{solon.pissis@cwi.nl}{https://orcid.org/0000-0002-1445-1932}{Supported by the PANGAIA and ALPACA projects that have received funding from the European Union’s Horizon 2020 research and innovation programme under the Marie Sk\l{}odowska-Curie grant agreements No 872539 and 956229, respectively.}
\author{Jakub Radoszewski}{University of Warsaw, Poland}{jrad@mimuw.edu.pl}{https://orcid.org/0000-0002-0067-6401}{Supported by the Polish National Science Center, grant no.\ 2022/46/E/{\allowbreak}ST6/00463.}
\author{Wojciech Rytter}{University of Warsaw, Poland}{rytter@mimuw.edu.pl}{https://orcid.org/0000-0002-9162-6724}{}
\author{Tomasz Wale\'n}{University of Warsaw, Poland}{walen@mimuw.edu.pl}{https://orcid.org/0000-0002-7369-3309}{}
\author{Wiktor Zuba}{CWI, Amsterdam, The Netherlands}{wiktor.zuba@cwi.nl}{https://orcid.org/0000-0002-1988-3507}{Received funding from the European Union's Horizon 2020 research and innovation programme under the Marie Skłodowska-Curie grant agreement Grant Agreement No 101034253.}

\acknowledgements{We thank Tomasz Kociumaka for helpful discussions.}

\EventEditors{Olaf Beyersdorff, Mamadou Moustapha Kant\'{e}, Orna Kupferman, and Daniel Lokshtanov}
\EventNoEds{4}
\EventLongTitle{41st International Symposium on Theoretical Aspects of Computer Science (STACS 2024)}
\EventShortTitle{STACS 2024}
\EventAcronym{STACS}
\EventYear{2024}
\EventDate{March 12--14, 2024}
\EventLocation{Clermont-Ferrand, France}
\EventLogo{}
\SeriesVolume{289}
\ArticleNo{40}

\authorrunning{P. Charalampopoulos et al.} 
\Copyright{Panagiotis Charalampopoulos, Solon P. Pissis, Jakub Radoszewski, Wojciech Rytter, Tomasz Waleń, Wiktor Zuba} 



\category{} 

\relatedversion{} 

\nolinenumbers 

\begin{document}

\maketitle

\begin{abstract}
In the $k$-Edit Circular Pattern Matching ($k$-Edit CPM) problem, we are given a length-$n$ text~$T$, a length-$m$ pattern $P$, and a positive integer threshold~$k$, and we are to report all starting positions of the substrings of $T$ that are at edit distance at most $k$ from some cyclic rotation of~$P$.
In the decision version of the problem, we are to check if any such substring exists.
Very recently, Charalampopoulos et al.~[ESA 2022] presented $\Oh(nk^2)$-time and $\Oh(nk \log^3 k)$-time solutions for the reporting and decision versions of $k$-Edit CPM, respectively. 
Here, we show that the reporting and decision versions of $k$-Edit CPM can be solved in $\Oh(n+(n/m)\, k^6)$ time and $\Oh(n+(n/m)\, k^5 \log^3 k)$ time, respectively, thus obtaining the first algorithms with a complexity of the type $\Oh(n+(n/m)\, \mbox{poly}(k))$ for this problem.
Notably, our algorithms run in $\cO(n)$ time when $m=\Omega(k^6)$ and are superior to the previous respective solutions when 
$m=\omega(k^4)$.
We provide a meta-algorithm that yields efficient algorithms in several other interesting settings, such as when the strings are given in a compressed form (as straight-line programs), when the strings are dynamic, or when we have a quantum computer.

We obtain our solutions by exploiting the structure of approximate circular occurrences of $P$ in~$T$, when $T$ is relatively short w.r.t. $P$.
Roughly speaking, either the starting positions of approximate occurrences of rotations of $P$ form $\cO(k^4)$ intervals that can be computed efficiently, or some rotation of $P$ is almost periodic (is at a small edit distance from a string with small period).
Dealing with the almost periodic case is the most technically 
demanding part of this work; we tackle it using properties of locked fragments (originating from [Cole and Hariharan, SICOMP 2002]).
\end{abstract}

\section{Introduction}
In the classic pattern matching (PM) problem, we are given
a length-$n$ text~$T$ and a length-$m$ pattern $P$, and we are to report all starting positions (called occurrences) of the fragments of $T$ that are identical to $P$.
This problem can be solved in the optimal $\Oh(n)$ time by, e.g., the famous Knuth-Morris-Pratt algorithm~\cite{DBLP:journals/siamcomp/KnuthMP77}.
In many real-world applications, we are interested in locating not only the fragments of $T$ which are identical to~$P$, but also the fragments of $T$ which are identical to any cyclic rotation of $P$.
In this setting, the rotations of $P$ form an equivalence class, represented by a single circular string. In the circular PM (CPM) problem, we are to report all occurrences of the fragments of $T$ that are identical to some cyclic rotation of $P$. The CPM problem can also be solved in $\Oh(n)$ time~\cite{DBLP:journals/corr/abs-2208-08915}.

Applications where circular strings are considered include the comparison of DNA sequences in bioinformatics~\cite{DBLP:journals/almob/GrossiIMPPRV16,BMCgenomicsAyad2017} as well as the comparison of shapes represented through directional chain codes in image processing~\cite{DBLP:journals/pr/Palazon-GonzalezMV14,DBLP:journals/prl/Palazon-Gonzalez15}.
In both applications, it is not sufficient to look for exact (circular) matches.
In bioinformatics, we need to account for DNA sequence divergence (e.g., in the comparison of different species or individuals); and in image processing, we need to account for small differences in the comparison of images (e.g., in classifying handwritten digits). This gives rise to the notion of edit distance on circular strings~\cite{DBLP:journals/ipl/Maes90,DBLP:journals/prl/AyadBP17}.

 We say that string $U$ is a \emph{(cyclic) rotation} of string $V$ if $U=XY$ and $V=YX$ for some strings $X$, $Y$, and write $V=\rot^i(U)$, where $i=|X|$; e.g., $U=\texttt{abcde},X=\texttt{ab},Y=\texttt{cde},V=\texttt{cdeab}=\rot^2(U)$. The \emph{edit (Levenshtein) distance} $\ed(U,V)$ of two strings $U$ and $V$ is the minimal number of letter insertions, deletions and substitutions required to transform $U$ to~$V$.
For two strings $U$ and $V$ and an integer $k>0$, we write $U =_k V$ if $\ed(U,V) \le k$ and 
we write $U \approx_k V$ if there exists a rotation $U'$ of $U$ such that $U' =_k V$.

For a string $U$ composed of letters $U[0],\ldots,U[|U|-1]$, by $U[i \dd j]=U[i \dd j+1)$ we denote the \emph{fragment of $U$} corresponding to the \emph{substring} $U[i] \cdots U[j]$.
We say that $T[p \dd p']$ is a \emph{circular $k$-edit occurrence} of pattern $P$ if $P \approx_k T[p \dd p']$.
By $\cycoc_k(P,T)$ we denote the set of starting positions of circular $k$-edit occurrences of $P$ in $T$.
Let us define $k$-Edit CPM (cf.\ \cref{fig:ex1}).

\defproblem{$k$-Edit CPM}{
  A text $T$ of length $n$, a pattern $P$ of length $m$, and a positive integer $k$.}{
  A representation of the set $\cycoc_k(P,T)$. ({\bf Reporting version})\\
  \hspace*{1.6cm}   Any position $i \in \cycoc_k(P,T)$, if there is any. ({\bf Decision version})
}

\begin{figure}[h]
\vspace*{-.5cm}
\centering
\begin{tikzpicture}
  \tikzstyle{red}=[color=red!90!black]
  \tikzstyle{darkred}=[color=red!50!black]
  \tikzstyle{blue}=[color=blue!50!black]
  \tikzstyle{black}=[color=black]
  \tikzstyle{green}=[color=green!30!black]
  \definecolor{turq}{RGB}{74,223,208}
  \definecolor{pink}{RGB}{254,193,203}

\begin{scope}[xshift=2.0cm]

    \begin{scope}[yshift=-1.4cm]
  \node at (-4.9,0) [left, above] {$T=$};
  \begin{scope}[xshift=-4.6cm]
  \foreach \c/\s [count=\i from 0] in {c/black,c/black,d/black,d/red,a/black,b/black,a/black,b/black,c/black} {
    \node at (\i * 0.3 + 0.3, 0) [above, \s] {\tt \c};
    \node at (\i * 0.3 + 0.3, 0.1) [below] {\tiny \i};
  }
  \end{scope}
  \end{scope}

  \begin{scope}[yshift=-0.7cm, xshift=-4.3cm]
  \node at (-0.7,-0.1) [above] {$\rot^2(P)=$};
  \draw [fill=green!20!white] (1.04, 0.07) rectangle (1.65, 0.38);
  \draw [fill=turq!20!white] (0.15, 0.07) rectangle (1.04, 0.38);
  \foreach \c/\s [count=\i from 0] in {c/blue,d/blue,-/red,a/green,b/green} {
    \node at (\i * 0.3 + 0.3, 0) [above, \s] {\tt \c};
  }
  \foreach \ii [count=\i from 0] in {2,3, ,0,1} {
    \node at (\i * 0.3 + 0.3, 0.1) [below] {\tiny \ii};
  }
  \end{scope}

\begin{scope}[yshift=-1.4cm]
  \node at (-0.3,0) [left, above] {$T=$};
  \begin{scope}[xshift=-0.3cm]
  \end{scope}
  \foreach \c/\s [count=\i from 0] in {c/black,c/black,d/red,d/black,a/black,b/black,a/black,b/black,c/black} {
    \node at (\i * 0.3 + 0.3, 0) [above, \s] {\tt \c};
    \node at (\i * 0.3 + 0.3, 0.1) [below] {\tiny \i};
  }
  \end{scope}

  \begin{scope}[yshift=-0.7cm, xshift=0.6cm]
  \node at (-0.7,-0.1) [above] {$\rot^2(P)=$};
  \draw [fill=green!20!white] (0.74, 0.07) rectangle (1.35, 0.38);
  \draw [fill=turq!20!white] (0.15, 0.07) rectangle (0.74, 0.38);
  \foreach \c/\s [count=\i from 0] in {c/red,d/blue,a/green,b/green} {
    \node at (\i * 0.3 + 0.3, 0) [above, \s] {\tt \c};
  }
  \foreach \ii [count=\i from 0] in {2,3,0,1} {
    \node at (\i * 0.3 + 0.3, 0.1) [below] {\tiny \ii};
  }
  \end{scope}

  \begin{scope}[yshift=-1.4cm]
  \node at (4.3,0) [left, above] {$T=$};
  \begin{scope}[xshift=4.6cm]
  \foreach \c/\s [count=\i from 0] in {c/black,c/black,d/black,d/black,a/black,b/black,a/red,b/black,c/black} {
    \node at (\i * 0.3 + 0.3, 0) [above, \s] {\tt \c};
    \node at (\i * 0.3 + 0.3, 0.1) [below] {\tiny \i};
  }
  \end{scope}
  \end{scope}

  \begin{scope}[yshift=-0.7cm, xshift=5.5cm]
  \node at (-0.7,-0.1) [above] {$\rot^3(P)=$};
  \draw [fill=green!20!white] (0.44, 0.07) rectangle (1.35, 0.38);
  \draw [fill=turq!20!white] (0.15, 0.07) rectangle (0.44, 0.38);
  \foreach \c/\s [count=\i from 0] in {d/blue,a/green,b/green,c/red} {
    \node at (\i * 0.3 + 0.3, 0) [above, \s] {\tt \c};
  }
  \foreach \ii [count=\i from 0] in {3,0,1,2} {
    \node at (\i * 0.3 + 0.3, 0.1) [below] {\tiny \ii};
  }
  \end{scope}

  \begin{scope}[yshift=-3.4cm]
  \node at (-4.9,0) [left, above] {$T=$};
  \begin{scope}[xshift=-4.6cm]
  \foreach \c/\s [count=\i from 0] in {c/black,c/black,d/black,d/black,a/black,b/red,a/black,b/black,c/black} {
    \node at (\i * 0.3 + 0.3, 0) [above, \s] {\tt \c};
    \node at (\i * 0.3 + 0.3, 0.1) [below] {\tiny \i};
  }
  \end{scope}
  \end{scope}

  \begin{scope}[yshift=-2.7cm, xshift=-3.1cm]
  \node at (-0.7,-0.1) [above] {$\rot^3(P)=$};
  \draw [fill=green!20!white] (0.44, 0.07) rectangle (1.35, 0.38);
  \draw [fill=turq!20!white] (0.15, 0.07) rectangle (0.44, 0.38);
  \foreach \c/\s [count=\i from 0] in {d/red,a/green,b/green,c/green} {
    \node at (\i * 0.3 + 0.3, 0) [above, \s] {\tt \c};
  }
  \foreach \ii [count=\i from 0] in {3,0,1,2} {
    \node at (\i * 0.3 + 0.3, 0.1) [below] {\tiny \ii};
  }
  \end{scope}

  \begin{scope}[yshift=-3.4cm]
  \node at (-0.3,0) [left, above] {$T=$};
  \begin{scope}[xshift=0cm]
  \foreach \c/\s [count=\i from 0] in {c/black,c/black,d/black,d/black,a/black,b/black,a/black,b/black,c/black,-/red} {
    \node at (\i * 0.3 + 0.3, 0) [above, \s] {\tt \c};
  }
    \foreach \c/\s [count=\i from 0] in {c/black,c/black,d/black,d/black,a/black,b/black,a/black,b/black,c/black} {
   \node at (\i * 0.3 + 0.3, 0.1) [below] {\tiny \i};
  }
  \end{scope}
  \end{scope}

  \begin{scope}[yshift=-2.7cm, xshift=1.8cm]
  \node at (-0.3,0) [above] {$P=$};
  \draw [fill=green!20!white] (0.15, 0.07) rectangle (1.35, 0.38);
  \foreach \c/\s [count=\i from 0] in {a/green,b/green,c/green,d/red} {
    \node at (\i * 0.3 + 0.3, 0) [above, \s] {\tt \c};
  }
  \foreach \ii [count=\i from 0] in {0,1,2,3} {
    \node at (\i * 0.3 + 0.3, 0.1) [below] {\tiny \ii};
  }
  \end{scope}

  \end{scope}  
\end{tikzpicture}
\caption{Illustration of the 1-edit circular occurrences of pattern $P=\mathtt{abcd}$ in text $T=\mathtt{ccddababc}$.  
We have $\cycoc_1(P,T)\,=\,\{1,2,3,5,6\}$. The letters involved in an edit operation are coloured red.}\label{fig:ex1}
\end{figure}
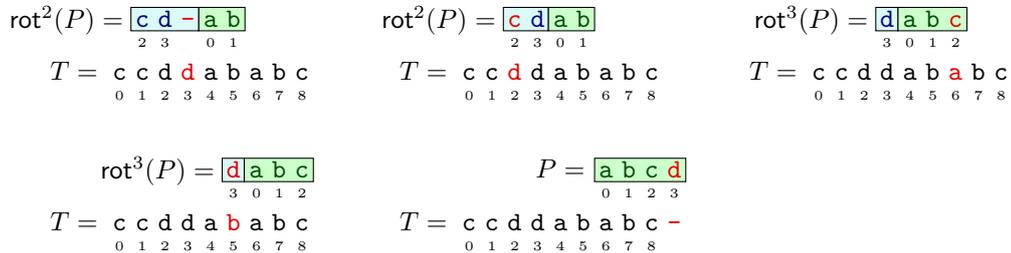
\vspace{-0.5cm}
\subparagraph*{Related work.}
The \emph{Hamming distance} of two equal-length strings $U$ and $V$ is the number of mismatches between $U$ and $V$; that is, the minimal number of letter substitutions required to transform $U$ to $V$.
Accounting for surplus or missing letters on top of substitutions poses significant challenges.
For example, the Hamming distance of two length-$n$ strings can be computed in $\cO(n)$ time with a trivial algorithm,
while it is known that their edit distance cannot be computed in $\cO(n^{2-\epsilon})$ time, for any $\epsilon>0$, under the Strong Exponential Time Hypothesis \cite{DBLP:journals/siamcomp/BackursI18}.
The situation is similar for (non-circular) approximate pattern matching.
The $k$-Mismatch PM problem is quite well-understood as the upper bound of $\cOtilde(n + k n /\sqrt{m})$
due to Gawrychowski and Uznański \cite{DBLP:conf/icalp/GawrychowskiU18},
who provided a smooth tradeoff between the algorithms of Amir et al.~\cite{DBLP:journals/jal/AmirLP04} with running time $\cOtilde(n\sqrt{k})$ and Clifford et al.~\cite{DBLP:conf/soda/CliffordFPSS16} with running time $\cOtilde(n+(n/m)k^2)$, is matched by a lower bound for so-called
``combinatorial'' algorithms.\footnote{Throughout this work, the $\cOtilde(\cdot)$ notation hides factors polylogarithmic in the length of the input strings.}
Algorithms that are faster by polylogarithmic factors have been presented in~\cite{DBLP:conf/stoc/ChanGKKP20,DBLP:conf/focs/Chan0WX23,FOCS20}. 
In contrast, the complexity of the $k$-Edit PM problem is not yet settled: the current records are the classic $\cO(nk)$-time algorithm of
Landau and Vishkin~\cite{DBLP:journals/jal/LandauV89} and the very recent $\cOtilde(n+(n/m)k^{3.5})$-time algorithm of Charalampopoulos et al.~\cite{FOCS22} improving the classic $\Oh(n+(n/m)k^4)$-time algorithm of Cole and Hariharan~\cite{DBLP:journals/siamcomp/ColeH02}.
However, there is no known lower bound for $k$-Edit PM ruling out an $\cO(n+(n/m)k^{2})$-time algorithm.

Recent results in pattern matching under both the Hamming distance and the edit distance for various settings~\cite{DBLP:journals/corr/abs-2402-07732,DBLP:conf/isaac/BathieKS23,DBLP:conf/soda/BringmannWK19,JCSS21,ESA22,FOCS20,FOCS22,DBLP:conf/cpm/CliffordGK0U22,DBLP:conf/soda/JinN23,DBLP:conf/focs/KociumakaPS21,DBLP:journals/corr/abs-2311-07415} were fuelled by a novel characterization of the structure of approximate occurrences. It is folklore knowledge that if $n\leq 3m/2$, either pattern~$P$ has a single exact occurrence in $T$ or both $P$ and the portion of~$T$ spanned by occurrences of $P$
are periodic (with the same period).
In 2019, Bringmann et al.~\cite{DBLP:conf/soda/BringmannWK19} showed that either $P$ has \emph{few} approximate occurrences (under the Hamming distance) or it is \emph{approximately periodic}.
Later, Charalampopoulos et al.~\cite{FOCS20} tightened this result and proved an analogous statement for approximate occurrences under the edit distance.

Let us now focus on approximate circular pattern matching.
The CPM problem under the Hamming distance is called the \emph{$k$-Mismatch CPM} problem. 
An $\Oh(nk)$-time algorithm and an $\cOtilde(n+(n/m)k^3)$-time algorithm were proposed for the reporting version of $k$-Mismatch CPM by Charalampopoulos et al.\ in \cite{JCSS21} and \cite{ESA22}, respectively, whereas an $\cOtilde(n+(n/m)k^2)$-time algorithm for its decision version was given in \cite{ESA22}.
Further, the authors of~\cite{DBLP:journals/almob/BartonIP14,DBLP:conf/wea/HirvolaT14a} presented efficient average-case algorithms for $k$-Mismatch CPM.
The $k$-Edit CPM problem was considered in~\cite{ESA22}, where an $\Oh(nk^2)$-time algorithm and an $\Oh(nk\log^3 k)$-time algorithm were presented for the reporting and decision version, respectively.
Until now, no algorithm with worst-case runtime $\Oh(n+(n/m)k^{\Oh(1)})$ was known for $k$-Edit CPM. Such an algorithm is superior over $\Oh(nk^{\Oh(1)})$-time algorithms when the number of allowed errors is small in comparison to the length of the pattern.
Here, we propose the first such algorithms.

\subparagraph*{Our result.}
In order to represent the output of our algorithm compactly, we need the notion of an \emph{interval chain}.
For two integer sets $A$ and $B$, let $A\oplus B = \{a + b\::\; a\in A,\, b \in B\}$.
We extend this notation for an integer $b$ to $A\oplus b = b\oplus A= A \oplus \{b\}$.
An interval chain for an interval $I$ and non-negative integers $a$ and $q$ is a set of the form \[\Chain(I,a,q)\,=\,I\,\cup\, (I\oplus q)\,\cup\, (I\oplus 2q)\,\cup \dots\cup\, (I\oplus aq).\]
Here $q$ is called the \emph{difference} of the interval chain.
For example the set of underlined intervals in \cref{June15} corresponds to $\Chain([3\dd 8],2,8)=[3\dd 8]\cup[11\dd 16]\cup[19\dd 24]$.  

Our main algorithmic result can be stated as follows (cf.\ \cref{tab:history}).

\begin{theorem}\label{thm:main}
The reporting version of the $k$-Edit CPM problem can be solved in $\Oh(n+(n/m)k^6)$ time, with the output represented as a union of $\Oh((n/m)k^6)$ interval chains.
The decision version of the $k$-Edit CPM problem can be solved in $\Oh(n+(n/m)k^5 \log^3 k)$ time.
\end{theorem}

The following notion of an \emph{anchor} (see also \cref{fig:ex}) is crucial for understanding the structure of (approximate) circular pattern matching.

\begin{definition}\label{def:anchored}
A circular $k$-edit occurrence $T[p\dd p']$ of $P$ is \emph{anchored} at position~$i$ (called  {\bf anchor}) if 
$\ed(T[p\dd i), Y) + \ed(T[i\dd p'], X) \leq k$, where $P=XY$ for some $X,Y$.
We denote 
\[\anchored_k(P,T,i)\,=\, \{\,p\,:\, T[p\dd p'] \ \mbox{is anchored at}\ i\ 
\mbox{for some}\ p'\,\}.\]
\end{definition}

\smallskip
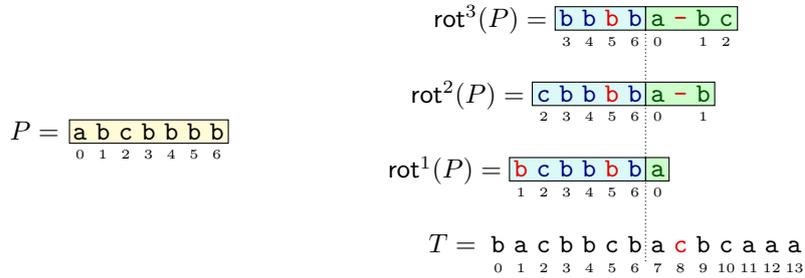
\begin{figure}[h]
\vspace*{-.4cm}
\centering
\begin{tikzpicture}
  \tikzstyle{red}=[color=red!90!black]
  \tikzstyle{darkred}=[color=red!50!black]
  \tikzstyle{blue}=[color=blue!50!black]
  \tikzstyle{black}=[color=black]
  \tikzstyle{green}=[color=green!30!black]
  \definecolor{turq}{RGB}{74,223,208}
  \definecolor{pink}{RGB}{254,193,203}

  \begin{scope}[xshift=-3.5cm,yshift=-0.5cm]
  \node at (-0.3,0) [left, above] {$P=$};
  \draw [fill=yellow!20!white] (0.16, 0.07) rectangle (2.25, 0.38);
  \foreach \c/\s [count=\i from 0] in {a/green,b/green,c/green,b/blue,b/blue,b/blue,b/blue} {
    \node at (\i * 0.3 + 0.3, 0) [above] {\tt \c};
    \node at (\i * 0.3 + 0.3, 0.1) [below] {\tiny \i};
  }
  \end{scope}

\begin{scope}[xshift=2.0cm]
  \begin{scope}[yshift=-2cm]
  \node at (-0.3,0) [left, above] {$T=$};
  \begin{scope}[xshift=-0.3cm]
  \end{scope}
  \foreach \c/\s [count=\i from 0] in {b/black,a/black,c/black,b/black,b/black,c/black,b/black,a/black,c/red,b/black,c/black,a/black,a/black,a/black} {
    \node at (\i * 0.3 + 0.3, 0) [above, \s] {\tt \c};
    \node at (\i * 0.3 + 0.3, 0.1) [below] {\tiny \i};
  }
  \end{scope}
  
  \begin{scope}[yshift=1cm, xshift=0.9cm]
  \node at (-0.7,-0.1) [above] {$\rot^3(P)=$};
  \draw [fill=green!20!white] (1.34, 0.07) rectangle (2.55, 0.38);
  \draw [fill=turq!20!white] (0.15, 0.07) rectangle (1.34, 0.38);
  \foreach \c/\s [count=\i from 0] in {b/blue,b/blue,b/red,b/blue,a/green,-/red,b/green,c/green} {
    \node at (\i * 0.3 + 0.3, 0) [above, \s] {\tt \c};
  }
  \foreach \ii [count=\i from 0] in {3,4,5,6,0,\,,1,2} {
    \node at (\i * 0.3 + 0.3, 0.1) [below] {\tiny \ii};
  }
  \draw[densely dotted] (1.34,0.07) -- (1.34,-3);
  \end{scope}
    
  \begin{scope}[yshift=0cm, xshift=0.6cm]
  \node at (-0.7,-0.1) [above] {$\rot^2(P)=$};
  \draw [fill=green!20!white] (1.64, 0.07) rectangle (2.55, 0.38);
  \draw [fill=turq!20!white] (0.15, 0.07) rectangle (1.64, 0.38);
  \foreach \c/\s [count=\i from 0] in {c/blue,b/blue,b/blue,b/red,b/blue,a/green,-/red,b/green} {
    \node at (\i * 0.3 + 0.3, 0) [above, \s] {\tt \c};
  }
  \foreach \ii [count=\i from 0] in {2,3,4,5,6,0,\,,1} {
    \node at (\i * 0.3 + 0.3, 0.1) [below] {\tiny \ii};
  }
  \end{scope}

    \begin{scope}[yshift=-1cm, xshift=0.3cm]
  \node at (-0.7,-0.1) [above] {$\rot^1(P)=$};
  \draw [fill=green!20!white] (1.94, 0.07) rectangle (2.25, 0.38);
  \draw [fill=turq!20!white] (0.15, 0.07) rectangle (1.94, 0.38);
  \foreach \c/\s [count=\i from 0] in {b/red,c/blue,b/blue,b/blue,b/red,b/blue,a/green,} {
    \node at (\i * 0.3 + 0.3, 0) [above, \s] {\tt \c};
  }
  \foreach \ii [count=\i from 0] in {1,2,3,4,5,6,0} {
    \node at (\i * 0.3 + 0.3, 0.1) [below] {\tiny \ii};
  }
  \end{scope}

  \end{scope}
  
\end{tikzpicture}
\caption{
The starting positions of circular 2-edit occurrences of pattern $P$ anchored at position 7 in text $T$ are $\anchored_2(P,T,7)=\{0,1,2,3,4\}$; the occurrences at positions $1,2,3$ are shown.
}\label{fig:ex}
\end{figure}

\begin{example}
Let $P=\mathtt{a}^{99}\,\mathtt{b}$ and $T=P^2$.
Then $|\cycoc_0(P,T)|=101$, while we have only two anchors (0 and 100).
\end{example}

Our algorithm exploits the approximate periodic structure of the two strings in scope. On the way to our main algorithmic result we prove 
(in the end of~\cref{sec:reduction}) the following structural result for $k$-Edit CPM:

\begin{restatable}{theorem}{thmcomb}\label{thm:comb}
Consider a pattern $P$ of length $m$, a positive integer threshold $k$, and a text~$T$ of length $n \leq cm+k$, for a constant $c \ge 1$.
Then, either there are only $\cO(k^2)$ anchors of circular $k$-edit occurrences of $P$ in $T$ or
some rotation of $P$ is at edit distance $\cO(k)$ from a string with period $\Oh(m/k)$.
\end{restatable}

\begin{table}
\begin{center}\scalebox{1.02}{
\begin{tabular}{||c c c || c c c||} 
 \hline
  $k$-\textbf{Edit} \textbf{PM} & \textbf{Reference} & \textbf{Note} & $k$-\textbf{Edit} \textbf{CPM} & \textbf{Reference} & \textbf{Note} \\  
 \hline\hline
  $\cO(n^2)$ & \cite{DBLP:journals/jal/Sellers80} & for any $k$ & $\cO(nk^2)$ & \cite{ESA22} &  \\
   $\cO(nk^2)$ & \cite{DBLP:journals/jcss/LandauV88} &  & $\ctO(nk)$ & \cite{ESA22} & decision \\
  $\cO(nk)$ & \cite{DBLP:journals/jal/LandauV89} &  & $\Oh(n+k^6 \cdot n/m)$ & \textbf{This work} & \\
   $\ctO(n+k^{\frac{25}{3}}\cdot n/m^\frac{1}{3})$ & \cite{DBLP:conf/focs/SahinalpV96} &  & $\ctO(n+k^5 \cdot n/m)$ & \textbf{This work} & decision  \\ 
    $\cO(n+k^{4}\cdot n/m)$ & \cite{DBLP:journals/siamcomp/ColeH02} &  & & & \\
    $\ctO(n+k^{3.5}\cdot n/m)$ & \cite{FOCS22} & & & & \\ 
 \hline
\end{tabular}}
\end{center}

\caption{The upper-bound landscape of pattern matching (PM) and circular PM (CPM) with $k$ edits. In the decision version of $k$-Edit CPM, the algorithms only find if there exists at least one occurrence and return a witness; otherwise the algorithms report all the occurrences.}
\label{tab:history}
\end{table}

\subparagraph*{The \pillar model.}
We work in the \pillar model that was introduced in~\cite{FOCS20} with the aim of unifying approximate pattern matching algorithms across different settings.
In this model, we assume that the following primitive \pillar operations
can be performed efficiently, where
the argument strings are fragments of strings in a given collection $\mathcal{X}$:
\begin{itemize}
\item $\Extract(S, \ell, r)$: Retrieve string $S[\ell\dd r)$.
\item $\LCP(S, T),\, \LCP_R(S, T)$: Compute the length of the longest common prefix/suffix of $S$, $T$.
\item $\IPM(S, T)$: Assuming that $|T| \le 2|S|$, compute the starting positions of all exact occurrences of $S$ in $T$, expressed as an arithmetic progression.
\item $\Access(S, i)$: Retrieve the letter $S[i]$; \
  $\Length(S)$: Compute the length $|S|$ of the string $S$.
\end{itemize}
The runtime of algorithms in this model can be expressed in terms of the number of primitive \pillar operations.
The result underlying \cref{thm:main} can be stated as follows.

\begin{theorem}\label{thm:main_pillar}
If $n \le m \le 2n$, the reporting and decision versions of the $k$-Edit CPM problem can be solved in $\cO(k^6)$ time and $\cO(k^5 \log^3 k)$ time in the \pillar model, respectively.
\end{theorem}

\cref{thm:main_pillar} implies \cref{thm:main} as well as efficient algorithms for $k$-Edit CPM in internal, dynamic, fully compressed, and quantum settings based on known implementations of the \pillar model in these settings, as discussed in Appendix~\ref{sec:results}.

\subparagraph*{Our approach.}
Every circular $k$-edit occurrence of $P$ in $T$ is anchored at some position $i$ of~$T$. In the reporting and decision version of the problem, we use the following respective results.

\begin{lemma}[{\cite[Lemma 30]{DBLP:journals/corr/abs-2208-08915}}]\label{lem:report-anchored}
Given a text $T$ of length $n$, a pattern $P$ of length $m$, an integer $k>0$, and a position $i$ of $T$, we can compute in $\Oh(k^2)$ time in the \pillar model the set $\anchored_k(P,T,i)$, represented as a union of $\cO(k^2)$ intervals, possibly with duplicates.
\end{lemma}

For an interval $I$ denote by $\anyanchored_k(P,T,I)$ an arbitrarily chosen  position in the set $\bigcup_{\,i\in I}\, \anchored_k(P,T,i)$; if this set is empty then the result is \emph{none}.
\begin{lemma}[{\cite[Section 4]{ESA22}}]\label{lem:decision-anchored}
Given a text $T$ of length $n$, a pattern $P$ of length $m$, an integer $k>0$, and an interval $I$ containing up to $k$ positions of $T$, we can compute  $\anyanchored_k(P,T,I)$ in $\Oh(k^2 \log^3 k)$ time in the \pillar model.
\end{lemma}

It will be convenient and sufficient  to deal separately with fragments of $T$ of length $\cO(m)$, so we can assume w.l.o.g.~that $n=\cO(m)$.
Let $P=P_1P_2$ be a decomposition of the pattern with $|P_1|=\floor{m/2}$.
By using \cref{lem:report-anchored} to compute $k$-edit circular occurrences that are anchored at one of $\cO(k^2)$ carefully chosen anchors, we reduce our problem to searching for $k$-edit (non-circular) occurrences of any length-$m$ substring of 
a certain fragment $V$ of $P_2P_1P_2$ in a suitable fragment $U$ of $T$, where both $V$ and $U$ are approximately periodic (there is also a symmetric case where $V$ is a substring of $P_1P_2P_1$).

We achieve this as follows.
Let us denote the set of standard (non-circular) \emph{$k$-edit occurrences} of a string $X$ in a string $Y$ by
$$\Occ_k(X,Y) = \{i \in [0\dd |Y|)\,:\,Y[i \dd i'] =_k X \ \text{for some}\ i' \ge i\}.$$
We compute the set $\Occ_k(P_1,T)$ using an algorithm for pattern matching with $k$ edits~\cite{FOCS20}.
If this set is small, it yields a small set of \emph{anchors} for $k$-edit occurrences of rotations of $P$ that contain $P_1$. 
We also do the same for $P_2$.
Then, we can apply \cref{lem:report-anchored} to each anchor. 

The challenging case is when $\Occ_k(P_1,T)$ is large.
The structural result for $k$-Edit PM then implies that $P_1$ and the portions of $T$ spanned by approximate occurrences of $P_1$ are \emph{almost periodic}, i.e., they are at small edit distance from a substring of string $Q^\infty$, where~$Q$ is a short string.
We extend the periodicity in each of $P_2P_1P_2$ and $T$, allowing for more edits.
The reduction is then completed by accounting for some technical considerations and, possibly, calling \cref{lem:report-anchored} $\cO(k^2)$ more times.

In order to develop some intuition for how to deal with the almost periodic case, let us briefly discuss how it is dealt with in the case where we are looking for approximate (circular) occurrences under the Hamming distance.
The mismatches of each of the two strings ($P$ and $T$ or $U$ and $V$) with a substring of $Q^\infty$ are called \emph{misperiods}.
Now, consider some candidate starting position $i$ of $P$ in $T$, assuming that both $P[0 \dd |Q|)$ and $T[i \dd i+|Q|)$ are approximate copies of $Q$: the number of mismatches of $P$ and $T[i \dd i+m)$ can be inferred by just looking at the misperiods: it is just the total number of misperiods in $P$ and $T[i \dd i+m)$ minus the misperiods that are aligned and thus ``cancel out''.

For approximate PM under the edit distance, the situation is much more complicated as deletions and insertions can be applied, and hence we cannot have an analogous statement about misperiods ``cancelling out''. Following works on (non-circular) $k$-edit PM, we employ so-called \emph{locked} fragments (see~\cite{FOCS20,DBLP:journals/siamcomp/ColeH02}).

Roughly speaking, we partition each of $U$ and $V$ into locked fragments and powers of~$Q$, such that the total length of locked fragments is small and, if a locked fragment is to be aligned with a substring of $Q^\infty$, we would rather align it with a power of $Q$.
Then, intuitively, one has to overcome technical challenges arising from the nature of the overlap of the locked fragments with a specific circular $k$-edit occurrence.

We consider different cases depending on whether the fragments of $U$ and $V$ that yield a match imply that any pair of locked fragments (one in $U$ and one in $V$) overlap.
A crucial observation is that, roughly speaking, as we slide a length-$m$ fragment of $V$ over $U$, $|Q|$ positions at a time, such that the locked fragments in the window in $U$ remain unchanged and do not overlap with locked fragments in $V$, the edit distance remains unchanged.

\section{Reduction of \texorpdfstring{$k$-Edit}{k-Edit} CPM to the PeriodicSubMatch Problem}
\label{sec:reduction}

A \emph{string} $S=S[0\dd |S|-1]$ is a sequence of \emph{letters} over some alphabet.
The string $S[i] S[i+1] \cdots S[j]$, for any indices $i,j$ such that $i \le j$, is called a \emph{substring} of $S$. By $S[i\dd j]=S[i\dd j+1)=S(i-1\dd j]$ we denote a \emph{fragment} of $S$ that can be viewed as a positioned substring $S[i] S[i+1] \cdots S[j]$ (it is represented in $\Oh(1)$ space).
We also denote $S^{(j)}=S[j \dd j+m)$.
An integer $p$ such that $0 < p \leq |S|$ is called a \emph{period} of $S$ if $S[i] = S[i + p]$, for all $i \in[0\dd |S| - p)$. We define \emph{the period} of $S$ as the smallest such $p$. 
A string $Q$ is called \emph{primitive} if $Q=W^k$ for a string $W$ and a positive integer~$k$ implies that $k=1$.
By $\rot^j(X)$ we denote the string $X[j\dd |X|)X[0\dd j)$.
We generalize the rotation operation $\rot$ to arbitrary integer exponents~$r$ as $\rot^r(X)=\rot^{r \bmod |X|}(X)$.

By $\edp{X}{Y}$, $\eds{X}{Y}$ and $\edl{X}{Y}$ we denote the minimum edit distance between string $X$
and any prefix, suffix and substring of string $Y^{|X|+|Y|}$, respectively.

We say that a string $U$ is \emph{almost $Q$-periodic} if $\edp{U}{Q}\leq 112 k$.
We write 
$a \equiv_d b \pmod{q}$ if  $a-b \equiv i \pmod{q},$
where $\min(i,q-i)\le d$ (in other words, $a$ and $b$ are $d$-approximately congruent modulo $q$). For example, $11 \equiv_{3} 21 \pmod{8}$, but $11 \equiv_{1} 21 \pmod{8}$ does not hold.

A pair of indices $(p,x)$ satisfying $p\in \Occ_k(V^{(x)},U)$ and $p\equiv_{77k} x+r \pmod{q}$ will be called an approximate match (\emph{app-match}, in short).

The following auxiliary problem, \EPSM, is illustrated in \cref{fig:setting}.
\defproblem{\bf \EPSM}{A primitive string $Q$, integers $m,r,k,\alpha,\beta$, and strings $U$, $V$ such that
\begin{itemize}
\item $m\leq |U| \leq \frac74m+3(k+1)$,\, $m\leq |V| \leq \frac32m$,\ $q=|Q|\leq \frac{m}{256\, k}$,\ $r \in [0\dd q)$,

\smallskip
\item $U$ is almost $Q$-periodic,

\smallskip
\item $V\,=\,P^2[\alpha\dd\beta]$ (hence, length-$m$ substrings of $V$ are rotations of $P$),

\smallskip
\item $V$ is almost $Q'$-periodic, where $Q':=\rot^r(Q)$.
\end{itemize}
} 
{$\{\,p\in \Occ_k(V^{(x)},U)\;:\;
p\equiv_{77k} x+r \pmod{q},\ x\le |V|-m\,\}$.
}

\begin{remark}
Due to the condition that $V$ is a fragment of $P^2$, we can apply the operation $\anchored_k$ to compute efficiently 
the output of \EPSM in the case when a position $j_1$ in $V$ is  \emph{aligned} with a position $i_1$ in $U$.
The efficiency of the whole approach is based on the efficiency of the operation $\anchored_k$.
\end{remark}

\begin{figure}[htpb]
    \centering
    \input{_fig_setting.tex}
    \caption{We have $m=25$, $k=2$ and $r=5$. Edits with respect to the approximate periodicity are marked in red. Green rectangles show that $V^{(x)}=_2 U[p \dd p+23)$. We have $p=x+r+1$, so $p \equiv_{1} x+r \pmod{q}$. The distances (in blue) from $p$ and $x$ to the starts of next approximate periods~$Q$ are the same up to $\Theta(k)$. For the example purposes, we waive the constraint $q=|Q|\leq \frac{m}{256\, k}$.
    }\label{fig:setting}
\end{figure}
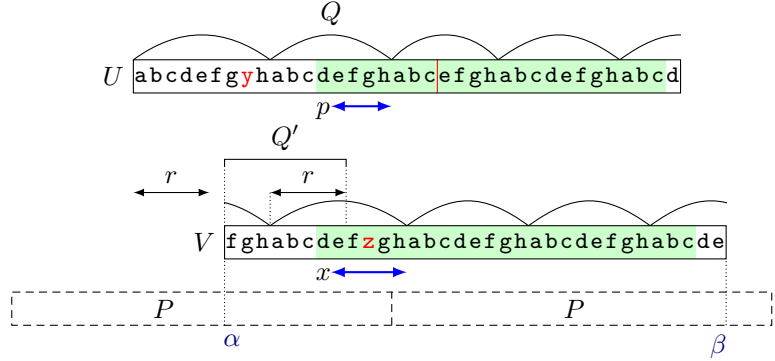

The strings $U$ and $V$ are both close to substrings of $Q^\infty$. The condition $p\equiv_{\Theta(k)} x+r \pmod{q}$ means that we are only interested in $k$-edit occurrences $U[p \dd p']$ of $V^{(x)}$ such that the two substrings are approximately synchronized with respect to the approximate period~$Q$; see \cref{fig:setting}.
(In particular, no other $k$-edit occurrences exist.)
The constants originate from \cref{mt} and some additional requirements imposed in the proof of \cref{lem:red}.

\begin{example}
A very simple \emph{double fully periodic} case, where both $U$ and $V$ are substrings of $Q^{\infty}$, is depicted in \cref{June15}. Again, we waive the constraint $q=|Q|\leq \frac{m}{256\, k}$.

\begin{figure}[htpb]
\centering
\input{_fig_uv.tex}
\caption{A double fully periodic case. Let $k=2$, $q=|Q|=8$, and $r=4$. For $m=23$, the set of $k$-edit occurrences of any length-$m$ fragment of $V$ (2 possibilities) in $U$ is the (underlined) interval chain. 
For $m=16$ it is a single interval. Position $x$ in $V$ is synchronized with respect to the periodicity
with any position $p$ in $U$ such that $p\equiv x+r \pmod{q} $.}\label{June15}
\end{figure}
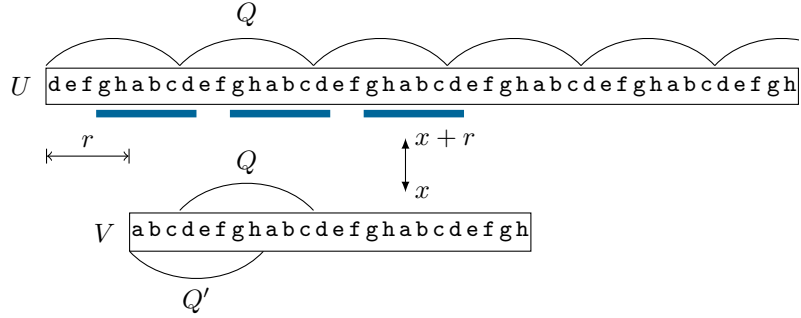
\end{example}

The following theorem follows as a combination of several results of \cite{FOCS20}, see below.

\begin{theorem}[\cite{FOCS20}]\label{mt}
  If $|T|=n < \frac32 m + k$, then in $\cO(k^4)$ time in the \pillar model we can compute 
   a representation of the set $\Occ_k (P, T )$.
  If $\floor{|\Occ_k (P, T )|/k} > 642045 \cdot (n/m) \cdot k$, the algorithm also returns:
  \begin{itemize}
  \item a primitive string $Q$ 
  satisfying 
  $|Q|\le m/(256k),\;\edl{P}{Q}=\edp{P}{Q} < 2k$,
  and 
  \item a fragment $\bar{T}$ of $T$ such that 
  $\edl{\bar{T}}{Q}\leq \edp{\bar{T}}{Q} \leq 24k$,
  $|\Occ_k (P,T)| = |\Occ_k (P,\bar{T})|$.
  \end{itemize}
  Moreover, $ i \equiv_{24k} 0 \pmod{|Q|}$
  for each $i \in \Occ_k(P,\bar{T})$.
\end{theorem} 

\subparagraph*{Origin of \cref{mt}.}
An algorithm that efficiently computes a representation of $\Occ_k (P, T )$ is encapsulated in \cite[Main Theorem 9]{FOCS20} \footnote{When referring to
statements of~\cite{FOCS20}, we use their numbering in the full (arxiv) version of the paper.}.
The first step of this algorithm is the analysis of the pattern specified in \cite[Lemma 6.4]{FOCS20}, which results in computing either a set of \emph{breaks}, a set of \emph{repetitive regions}, or a primitive string $Q$ that is of length at most $m/(128k)$ and satisfies $\edl{P}{Q} < 2k$.
In the presence of breaks or repetitive regions, we have $\floor{|\Occ_k (P, T )|/k} \le 642045 \cdot (n/m) \cdot k$, see \cite[Lemmas 5.21 and 5.24]{FOCS20}.
In the case where the analysis of the pattern returns an approximate period~$Q$, we can use \cite[Lemma 6.5]{FOCS20} to find a rotation $Q_1$ of $Q$ such that $\edl{P}{Q_1}=\edp{P}{Q_1}$.
Set $Q:= Q_1$.
Now, let us also compute all $k$-edit occurrences of the reversal of $P$ in the reversal of $T$.
Then, we can trim $T$, obtaining a string $\bar{T}$ so that all $k$-edit occurrences of $P$ in $T$ are preserved in $\bar{T}$, and $P$ has a $k$-edit occurrence both as a prefix and as a suffix of~$\bar{T}$.
We can then directly apply \cite[Theorem 5.2]{FOCS20} with $d=8k$ to obtain the stated properties of~$\bar{T}$; for the fact that $\edp{\bar{T}}{Q} \leq 24k$ holds see the fourth paragraph of the proof of that theorem.
The length of~$Q$ can be instead bounded by $m/(256k)$ with all other constants remaining unchanged; this is because the bottleneck for the number of occurrences in the case where $P$ is not almost periodic stems from repetitive regions and is not sensitive to the exact length of $Q$.
That is, it is only the number of occurrences in the case where the analysis of the pattern yields $2k$ breaks that can be larger (by a multiplicative factor of 2), but the bound stated above is dominant.

\begin{remark}
An $\Oh((n/m)k^{3.5}\sqrt{\log k \log m})$-time algorithm for computing a representation of the set $\Occ_k(P,T)$ using $\Oh(k^3)$ arithmetic progressions was presented in \cite{FOCS22}. The simpler result from \cite{FOCS20} is sufficient for our needs.
\end{remark}
    The proof of the following \cref{lem:red} resembles the proof of \cite[Lemma 12]{ESA22} which is an analogous fact stated for the Hamming distance. We use the fact that \cite{FOCS20} provides a unified framework for the two metrics, but still need to overcome the technical difficulties that arise from replacing Hamming distance with edit distance.

    \begin{lemma}\label{lem:red}
    If $n=\Oh(m)$, then 
    $k$-Edit CPM can be reduced in $\Oh(k^4)$ time in the \pillar model to at most two instances of the \EPSM problem.
    The output to $k$-Edit CPM is a union of the outputs of the two \EPSM instances and $\Oh(k^4)$ intervals.
    \end{lemma}
    
    Before we proceed with the proof, let us recall some notions and their properties from \cite{FOCS20}.
    One of our main tools are repetitive regions. Intuitively, repetitive regions are fragments that are approximately highly periodic and at the same time, they have a given large number of edits with respect to the periodicity.

\begin{definition}\label{def:repreg_edit}
We say that a fragment $R$ of a string $S$ of length $m$ is a \emph{repetitive region} if $|R| \geq 3m/8$ and there is a primitive string $Q$ such that $|Q| \leq m/(256k)$ and $\edl{R}{Q} = \ceil{8k|R|/m}$.
\end{definition}

\begin{lemma}[see~{\cite[Lemma 5.24]{FOCS20}}]\label{repreg}
Consider a pattern $P$ of length $m$, a text $T$ of length~$n$ and a positive integer threshold $k \le m$.
If the pattern $P$ contains a repetitive region, then $|\Occ_k(P,T)|=\Oh((n/m)k^2)$.
\end{lemma}

We also use the following two auxiliary lemmas from \cite{FOCS20}. Intuitively, \cref{lem:misope} adapts the Landau-Vishkin algorithm (which can be viewed as a generalization of kangaroo jumps to the edit distance).

\begin{lemma}[{\cite[Lemma 6.1]{FOCS20}}]\label{lem:misope}
    Let $S$ denote a string and let $Q$ denote a string (that is possibly given as
    a cyclic rotation $\rot^j(Q)$).
    Then, there is a generator {\tt EditGenerator($S$, $Q$)} ({\tt EditGenerator$^R$($S$, $Q$)}) that
    in the $k$-th call to
    {\tt Next}, returns in $\Oh(k)$ time in the \pillar model the length of the longest prefix (suffix) $S'$ of $S$ and the length of
    the corresponding prefix (suffix) $Q'$ of $Q^\infty$ such that $\ed(S', Q') \le k$.
\end{lemma}

\begin{lemma}[{\cite[Lemma 6.3]{FOCS20}}]\label{clm:fixed}
Let $S$ be a string such that
\[|S| \ge (2t+1)|Q|\ \text{and} \edl{S}{Q} = \edp{S}{(Q')} = \eds{S}{(Q'')} \le t,\]
where $Q'$ and $Q''$ are rotations of a string $Q$.

If $\edl{SS'}{Q} \le t$ for a string $S'$, then $\edl{SS'}{Q}=\edp{SS'}{(Q')}$.

If $\edl{S'S}{Q} \le t$ for a string $S'$, then $\edl{S'S}{Q}=\eds{S'S}{(Q'')}$.
\end{lemma}

\begin{proof}[Proof of \cref{lem:red}]
Let us partition $P$ to two (roughly) equal chunks, $P_1$ of length $\floor{m/2}$ and $P_2$ of length $\ceil{m/2}$. Each circular 
$k$-edit occurrence of $P$ in $T$ implies a standard
$k$-edit occurrence of at least one of $P_1$ and $P_2$.
We focus on the case when it implies such an occurrence of $P_1$, noting that the computations for $P_2$ are symmetric.

For a fragment $T'$ of $T$, we denote by $\implied_k(P_1,T')$ the set of circular $k$-edit occurrences of $P$ in $T$ in which a $k$-edit occurrence of $P_1$ is contained in $T'$; a formal definition follows.

\begin{definition}\label{def:implied2}
For $T'=T[t \dd t']$, we define $\implied'_k(P_1,T')$ as a set of pairs $(j,y)$ such that $j \in [0\dd n)$, $m/2 \le y < m$, and there exist positions $i,i',j' \in [0 \dd n)$ and such that:
\begin{itemize}
\item $[i \dd i'] \subseteq [t \dd t']$ and
\item $\gamma := \ed(T[j \dd i),P[y \dd m)) + \ed(T[i \dd i'],P_1) + \ed(T(i' \dd j'],P[|P_1| \dd y)) \le k$.
\end{itemize}
Then, $\implied_k(P_1,T')=\{j\,:\,(j,y) \in \implied'(P_1,T')\}$.
\end{definition}

Let us note that in the above definition $\ed(T[j \dd j'],\,\rot^y(P)) \le \gamma$.

We cover $T$ with fragments of length $\floor{\frac32 |P_1|} + k$ starting at multiples of $\floor{\frac12|P_1|}$. (The last fragments can be shorter.) For each of the fragments $T'$ of $T$, we will compute a representation of a set $A$ such that $\implied_k(P_1,T') \subseteq A \subseteq \cycoc_k(P,T)$.
If $|\Occ_k (P_1, T')|=\Oh(k^2)$, we use the following fact whose proof is based on anchors.

\begin{claim}\label{June20}
If the set $\Occ_k (P_1, T')$ for a fragment $T'$ of $T$ has size $\Oh(k^2)$ and is given,
then a set of positions $A$ such that $\implied_k(P_1,T') \subseteq A \subseteq \cycoc_k(P,T)$, represented as a union of $\Oh(k^4)$ intervals, can be computed in $\Oh(k^4)$ time in the \pillar model.
\end{claim}
\begin{proof}
We compute the set $\anchored_k(P,T,i+s)$ 
for each position $i \in \Occ_k(P_1,T')$, where~$s$ is the starting position of $T'$ in $T$ (i.e., $T'=T[s \dd s+|T'|)$).
By \cref{lem:report-anchored}, this set is represented as a union of $\Oh(k^2)$ intervals and can be computed in $\Oh(k^2)$ time in the \pillar model. Since $|\Occ_k(P_1,T')|=\Oh(k^2)$, the union $A$ of all these sets contains $\cO(k^4)$ intervals and is computed in $\Oh(k^4)$ total time.
Clearly, $A$ satisfies the required inclusions.
\end{proof}

If $|\Occ_k (P_1, T')|=\Oh(k^2)$, \cref{June20} produces such a representation consisting of $\Oh(k^4)$ intervals in $\Oh(k^4)$ time in the \pillar model.
Henceforth we assume that $\floor{|\Occ_k (P_1, T')|/k} > 642045 \cdot (|T'|/|P_1|) \cdot k$. In this case, by \cref{mt}, $P_1$ and the relevant part $\bar{T}'$ of~$T'$ are both almost $Q$-periodic. More formally, the algorithm behind the theorem returns a short primitive string $Q$ and a 
fragment $\bar{T'}$ of $T'$ that contains all occurrences of $P_1$ in $T'$ such that $\edp{P_1}{Q} < 2k$ and $\edp{\bar{T}'}{Q} \le 24k$.
We will compute strings $V$ and $U$ being fragments of $P_2P_1P_2$ and $T$, respectively, and obtain the required set $A$ as a union of the answer to \EPSM for $U$ and $V$ and $\Oh(k^4)$ intervals of positions.

\begin{figure}[htpb]
\vspace*{-0.3cm}
\centering
\begin{tikzpicture}[yscale=0.4]
    \draw (0,0) rectangle (9,1);
    \draw (3,0) -- (3,1)  (6,0) -- (6,1);
    \draw (1.5,0.5) node {\small $P_2$};
    \draw (4.5,0.5) node { ${P_1}$};
    \draw (7.5,0.5) node {\small $P_2$};
    \draw[thick,blue!70!black] (3,1.2) -- (3,1.4) -- node[above] {$R_{\rright}$} (8,1.4) -- (8,1.2);
    \draw[thick,green!30!black] (2,-0.2) -- (2,-0.4) -- node[below] {$R_{\lleft}$} (6,-0.4) -- (6,-0.2);
    \begin{scope}[yshift=-0.5cm]
    \draw[red!80!black,very thick] (2,-1) -- (2,-1.2) -- (8,-1.2) -- (8,-1);
    \draw[red!50!black] (8,-1.1) node[right] {$V$};
    \end{scope}
    \draw[densely dotted] (2,-1.5) -- (2,-0.4)  (8,-1.5) -- (8,1.2);
\end{tikzpicture}
\caption{String $V$ (shown in brown) and repetitive regions $R_{\lleft}$ and $R_{\rright}$ in $P_2P_1P_2$.
}\label{fig:rr}
\end{figure}
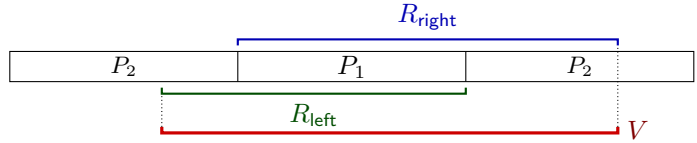

\proofsubparagraph{Computing $V$.}
Intuitively, string $V$ is computed by extending the approximate periodicity of the middle fragment $P_1$ in $P_2P_1P_2$ towards both directions. (Note that all rotations of $P$ that contain its first half $P_1$ are substrings of $P_2P_1P_2$.)
In each direction, we stop extending when either $c'k$ errors to a prefix (suffix) of $Q^{|V|}$ are accumulated, for a specified constant $c'$, or we reach the end of the string.
In the former case, we obtain a repetitive region $R_{\rright}$ with a prefix $P_1$ ($R_{\lleft}$ with a suffix $P_1$, respectively); see \cref{fig:rr}.

\begin{center}
\fbox{
\begin{minipage}{9cm}
\vspace*{0.1cm}
\hspace*{-0.2cm}  {\bf function} ComputeV$(P)$

\smallskip
compute $Q$ (\cref{mt})

let $Z=W=P_1$ be the occurrence of $P_1$ in $P_2P_1P_2$

\smallskip
Use \texttt{EditGenerator} to extend $W$ to the right until at least one 
of the following two conditions is satisfied: 
\vspace*{2mm}
\begin{description}
    \item{\bf (a)} 
$W$ is a repetitive region w.r.t.\ $Q$;
    \item{\bf 
(b)} we reach the end of $P_2P_1P_2$.
\end{description}

{\bf if} $W \ne P_1P_2$ {\bf then} $R_{\rright}:=W$

\smallskip
Use \texttt{EditGenerator}$^R$ to extend $Z$ to the left until any of \\
the following two conditions is satisfied:
\vspace*{2mm}
\begin{description}
\item{\bf 
(a)}  we reach the beginning of $P_2P_1P_2$; 
 \item{\bf (b)} $Z$ is a repetitive region w.r.t.\ $Q$.
\end{description}

{\bf if} $Z \ne P_2P_1$ {\bf then} $R_{\lleft}:=Z$

$V:=$\,the shortest substring of $P_2P_1P_2$ containing $Z$ and $W$

\smallskip
{\bf return} $V$, $R_{\lleft}$, $R_{\rright}$
\vspace*{0.1cm}
\end{minipage}
}
\end{center}

More precisely, in the function $ComputeV(P)$ we first extend the fragment equal to $P_1$ to the right, trying to accumulate enough errors with a prefix of $Q^\infty$ in order to reach the threshold specified in \cref{def:repreg_edit}, which is $\Theta(k)$.
Initially, $\edl{P_1}{Q}=\edp{P_1}{Q} < 2k < \ceil{8k|P_1|/m}$.
We use a technique that was developed to compute repetitive regions in the proof of \cite[Lemma 6.4]{FOCS20}.
In short, the \texttt{EditGenerator} from \cref{lem:misope} allows us to find in $\Oh(k)$ time, for each (subsequent) value $\delta$, the longest prefix $W$ of $P_1P_2$ such that $\edp{W}{Q} \le \delta$ until either the threshold from \cref{def:repreg_edit} is reached or $W=P_1P_2$ and we conclude that $\edp{P}{Q}<8k$. By \cref{clm:fixed} and the fact that $q<m/(256k)$, we then have $\edl{W}{Q}=\edp{W}{Q}=\delta$.

We perform the same process by extending the specified occurrence of $P_1$ to the left, using \texttt{EditGenerator}$^R$, obtaining a substring $Z$.
Formally, we first use the \texttt{EditGenerator} for $P_1$ and $Q$ to infer in $\Oh(k^2)$ time a prefix $Q^\infty[0 \dd y)$ of $Q^\infty$ such that $\edp{P_1}{Q}=\ed(P_1,Q^\infty[0 \dd y))$. Then, we apply \texttt{EditGenerator}$^R$ to suffixes of $P_2P_1$ and rotation $\rot^y(Q)$.

We let $V=P^2[\alpha \dd \beta]$ be the shortest substring of $P_2P_1P_2$ that spans both $Z$ and $W$.
By \cref{def:repreg_edit}, we have $\edl{V}{Q} \le 2\cdot \ceil{8km/m}=16k$.

Thus a rotation of $P$ that contains~$P_1$ either contains one of the repetitive regions or it is contained in $V$. By \cref{repreg}, a repetitive region has $\Oh(k^2)$ occurrences in a string of length $\Oh(m)$, so the former case can be solved in $\Oh(k^4)$ time with the aid of anchors as in \cref{June20}. The latter case will lead to \EPSM.

\proofsubparagraph{Computing $U$.}
Similarly, the function $ComputeU(T)$ computes $U$ as an extension of~$\bar{T}'$.
For reasons that will become apparent in the proof of \cref{clm:rots} our stopping conditions on accumulating edits are slightly different.
In the extension to left, we keep going until we have a substring $G$ such that $\eds{G}{(\rot^x(Q))} \geq 10k$ for all $x \in [-34k \dd 34k]$. The extension to the right is similar.
The total time required for this extension in the \pillar model is $\cO(k^3)$.
We will show that the resulting substring $U$ contains all the remaining elements of $\implied_k(P_1,T')$ and that the approximate congruence is satisfied.

\begin{center}
\fbox{
\begin{minipage}{10cm}
\vspace*{0.1cm}
\hspace*{-0.2cm}  {\bf function} ComputeU$(T)$

\smallskip
initially $U=\bar{T}'$;\ compute $Q$ (\cref{mt})

\smallskip
Use \texttt{EditGenerator} to extend $U$ to the right until at least one of the following three conditions is satisfied: 
\vspace*{2mm}
\begin{description}
    \item{\bf 
(a)} we reach the end of $T$; 
    \item{\bf (b)} 
 we have appended $|P_2|+k$ letters; or 
    \item{\bf 
(c)} the appended fragment $G$ of $T$ satisfies 

\hspace*{4mm} $\edp{G}{(\rot^x(Q))} \geq 10k$ for all $x \in |\bar{T}'|\oplus [-34k \dd 34k]$.
\end{description}

\smallskip
Use \texttt{EditGenerator}$^R$ to extend $U$ to the left until any of the following three conditions is satisfied: 
\vspace*{2mm}
\begin{description}
\item{\bf 
(a)}  we reach the beginning of $T$; 
 \item{\bf (b)} we have prepended $|P_2|+k$ letters; or 
  \item{\bf (c)} the prepended substring $F$ of $T$ satisfies
  
  \hspace*{4mm} $\eds{F}{(\rot^x(Q))} \geq 10k$ for all $x \in [-34k \dd 34k]$.
\end{description}

\smallskip
{\bf return} $U$
\vspace*{0.1cm}
\end{minipage}
}
\end{center}

\proofsubparagraph{Verifying approximate congruences.}
Intuitively, the approximate congruence $\bmod |Q|$ in \EPSM follows from the analogous condition in \cref{mt}.

Let $X$ be the substring of $U$ such that $U$ has a prefix $X\bar{T}'$ and $Y$ be the substring of $V$ such that $V$ has a prefix $YP_1$ (i.e., $Z=YP_1$). Both substrings $X$, $Y$ were computed using \texttt{EditGenerator}$^R$, which also produced suffixes $X'$ and $Y'$ of $Q^{2n}$ such that $\ed(X,X') \le 10k+34k=44k$ and $\ed(Y,Y') \le 8k$; see \cref{fig:xyxy}. 
Let $Q_1=\rot^{-|X'|}(Q)$ and $Q_2=\rot^{-|Y'|}(Q)$. We have $\edp{U}{Q_1} \le 2\cdot 44k + 24k=112k$ and $\edp{V}{Q_2} \le 16k$, so $U$ and $V$ are approximately $Q_1$-periodic and approximately $Q_2$-periodic, respectively. In \EPSM, we can therefore take $Q=Q_1$, $Q'=Q_2$ and $r=(|X'|-|Y'|) \bmod q$.

\begin{figure}[htpb]
\centering
\begin{tikzpicture}[xscale=0.55,yscale=0.38]
    \draw[thick,blue] (-3,0) rectangle node {$T$} (12,1);
    \draw[thick] (0,1) rectangle node {$U$} (10.5,2);
    \draw[thick] (0,2) rectangle (10,3);
    \draw[thick] (3,2) -- (3,3);
    \draw (1.5,2.5) node {$X$};
    \draw (6.5,2.5) node {$\bar{T}'$};
    \draw (-0.2,3) rectangle node {$X'$} (3,4);
    \begin{scope}
    \clip (-0.2,4) rectangle (3,5);
    \foreach \dx in {-1.5,0,1.5}{
        \draw[xshift=\dx cm] (0,4) .. controls (0.5,5) and (1,5) .. (1.5,4);
    }
    \end{scope}
    \draw (2.25,4.5) node[above] {$Q$};
    \draw (-0.2,5) rectangle node {$Q_1$} (1.3,6);
    \draw[densely dotted] (-0.2,5) -- (-0.2,4)  (1.3,5) -- (1.3,4);

    \draw[densely dotted] (0,1) -- (0,-5);

    \begin{scope}[yshift=-7cm]
    \draw[thick,brown!70!black] (-1,0) rectangle (14,1);
    \draw[thick,green!20!black] (4,0) -- (4,1)  (9,0) -- (9,1);
    \draw[blue] (1.5,0.5) node {$P_2$};
    \draw[blue] (6.5,0.5) node {$P_1$};
    \draw[blue] (11.5,0.5) node {$P_2$};
    \draw[thick] (0,1) rectangle node {$V$} (11,2);
    \draw[thick] (0,2) rectangle node {$Y$} (4,3);
    \draw[densely dotted] (4,1) -- (4,2);

    \draw (0.2,3) rectangle node {$Y'$} (4,4);
    \begin{scope}
    \clip (0.2,4) rectangle (4,5);
    \foreach \dx in {-0.5,1,2.5}{
        \draw[xshift=\dx cm] (0,4) .. controls (0.5,5) and (1,5) .. (1.5,4);
    }
    \end{scope}
    \draw (3.25,4.5) node[above] {$Q$};
    \draw (0.2,5) rectangle node {$Q_2$} (1.7,6);
    \draw[densely dotted] (0.2,5) -- (0.2,4)  (1.7,5) -- (1.7,4);
    \end{scope}
\end{tikzpicture}
\caption{Definitions of strings $X$, $X'$, $Y$, $Y'$, $Q_1$, $Q_2$.}\label{fig:xyxy}
\end{figure}
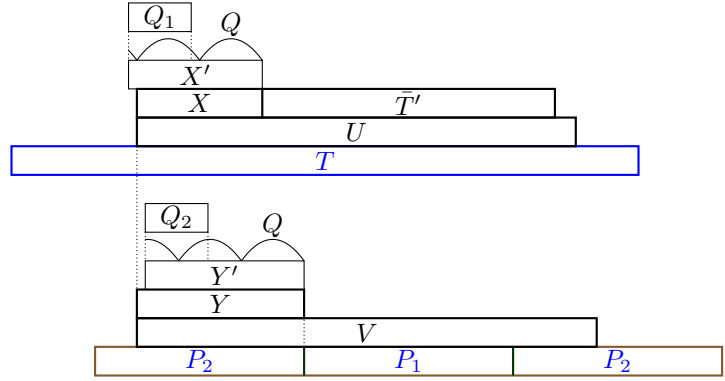

Let $(j,y) \in \implied'_k(P_1,T')$ and let us use the notations $i,i',j'$ from \cref{def:implied2}. Let $\bar{i},\bar{i}'$ be the positions in $\bar{T}'$ such that the fragments $T[i \dd i']$ and $\bar{T}'[\bar{i}\dd \bar{i}']$ correspond.
Assume that positions $p$, $p'$, $x$ satisfy 
$U[p \dd p']=T[j \dd j']\text{ and }\rot^y(P)=V^{(x)}$.
Then 
\[\ed(U[p \dd |X|+\bar{i}),V[x \dd |Y|)) \le k, \text{ so } |X|+\bar{i}-p \equiv_k |Y|-x \pmod{q}.\] 
\noindent 
By \cref{mt}, $\bar{i} \equiv_{24k} 0 \pmod{q}$. This fact and definitions of $X,Y$ lead to the following sequence of approximate congruences:
 \[p-x \equiv_k |X|+\bar{i}-|Y| \equiv_{25k} |X|-|Y| \equiv_{77k} |X'|-|Y'| \pmod{q}\]
so $(|X'|-|Y'|) \bmod q = r$, by definition.

\proofsubparagraph{Computing $A$.} We show how to compute a representation of a set $A$ such that \linebreak$\implied_k(P_1,T') \subseteq A \subseteq \cycoc_k(P,T)$.

The rotations of $P$ that contain $P_1$ are in one-to-one correspondence with the length-$m$ substrings  of $P_2P_1P_2$.
Each such substring contains $R_{\lleft}$, contains $R_{\rright}$, or is contained in~$V$.
We first show that we can efficiently compute circular $k$-edit occurrences of $P$ that imply $k$-edit occurrences of either $R_{\lleft}$ or $R_{\rright}$ (if they exist) using \cref{lem:report-anchored}.
We focus on $R_{\rright}$ as $R_{\lleft}$ can be handled symmetrically.
Due to \cref{repreg}, $R_{\rright}$ has $\cO(k^2)$ $k$-edit occurrences in $T$,
and they can be found in $\Oh(k^4)$ time in the \pillar model using \cref{mt}.
For each such occurrence at position $i$, we perform a call to $\anchored_k(P,T,i)$
that takes $\cO(k^2)$ time, for a total of $\cO(k^4)$ time in the \pillar model. We obtain $\Oh(k^4)$ intervals of positions contained in $\cycoc_k(P,T)$.

Now we focus on the remaining elements of the set $\implied_k(P_1,T')$ that correspond to $k$-edit occurrences of length-$m$ substrings of $V$ in $T$. By the following claim, it suffices to restrict the search to occurrences in $U$.

\begin{claim}\label{clm:rots}
If $j\in \implied_k(P_1,T')$, $U=T[u \dd u']$, and $T[j \dd j']$ is the substring at distance at most $k$ from a length-$m$ substring of $V$ defined as in \cref{def:implied2}, then $[j \dd j'] \subseteq [u \dd u']$.
\end{claim}
\begin{proof}
The proof is by contradiction.
Suppose $[j \dd j'] \not\subseteq [u \dd u']$.

We first assume that $j<u$.
We notice that $j<u$ is possible only if we stopped extending~$U$ to the left because we accumulated enough errors.
Assume that $(j,y) \in \implied'_k(P_1,T')$ and let $\rot^y(P)=V^{(x)}$ and positions $i,i'$ be defined as in \cref{def:implied2} (in particular, $T[i \dd i']$ is aligned with $P_1$).
As before, let $X$ be the substring of $U$ such that $U$ has a prefix $X\bar{T}'$ and $Y$ be the substring of $V$ such that $V$ has a prefix $YP_1$.
We have $V[x \dd |Y|) = P[y \dd m)$.
Further, let $F$ be a suffix of $Q^{2n}$ for which $\ed(V[x \dd |Y|),F)$ is minimal. By how~$V$ was computed towards identifying a repetitive region, we have $\ed(V[x \dd |Y|),F) \le 8k$.
Then, we have that $\ed(T[j \dd i),F)$ equals the minimum of $\ed(T[j \dd u+|X|),F_1) + \ed(T[u+|X| \dd i),F_2)$ over all partitions $F=F_1F_2$.
Now, if the second summand is less than $10k$, and since $i - (u + |X|) \equiv_{24k} 0 \pmod q$, $F_1$ is equal to a suffix of $(\rot^y(Q))^{2n}$ for some $y\in [-34k \dd 34k]$. Hence, since $j<u$, the computation of $U$ guarantees that $\ed(T[j \dd i),F) \ge 10k$.

Then, via the triangle inequality, we have
\begin{align*}
k &\geq \ed(\rot^y(P), T[j \dd j']) \\ 
& \geq \ed(P[y\dd m),T[j\dd i))
 = \ed(V[x\dd |Y|),T[j\dd i))\\
& \geq  \ed(T[j\dd i), F) - \ed(F, V[x\dd |Y|))\geq 10k - 8k>k,
\end{align*}
thus obtaining a contradiction.

Now assume that $j'>u'$.
We have $V[|Y| \dd x+m) = P[0 \dd y)$.
Let $F$ be a prefix of $Q^\infty$ for which $\ed(V[|Y| \dd x+m),F)$ is minimal. By how~$V$ was computed towards identifying a repetitive region, we have $\ed(V[|Y| \dd x+m),F) \le 8k$.
Then, we have that $\ed(T[i \dd j'],F)$ equals the minimum of $\ed(T[i \dd u'+|X|],F_1) + \ed(T(u'+|X| \dd j'],F_2)$ over all partitions $F=F_1F_2$.
Now, if the first summand is less than $10k$, and since $(u' + |X|) - i\equiv_{24k} |\bar{T}'| \pmod q$, $F_2$ is equal to a prefix of $(\rot^y(Q))^{2n}$ for some $y\in |\bar{T}'| \oplus [-34k \dd 34k]$. Hence, since $j'>u'$, the computation of $U$ guarantees that $\ed(T[i \dd j'],F) \ge 10k$.

Then, via the triangle inequality, we have
\begin{align*}
k &\geq \ed(\rot^y(P), T[j \dd j']) \\ 
& \geq \ed(P[0\dd y),T[i\dd j'])
 = \ed(V[|Y| \dd x+m),T[i\dd j'])\\
& \geq  \ed(T[i\dd j'], F) - \ed(F, V[|Y| \dd x+m))\geq 10k - 8k>k,
\end{align*}
thus obtaining a contradiction and completing the proof of the claim.
\end{proof}

Therefore, by the last claim, the remaining elements of the set $A$ are included in the output to \EPSM for $U$ and $V$. If $|U|<m$, we do not need to construct the instance of \EPSM.
This also completes the proof of the lemma.
\end{proof}

Let us now restate and prove our structural result.

\thmcomb*
\begin{proof}
\cref{thm:comb} readily follows from the proof of \cref{lem:red}. If $|V|\geq m$, then some rotation of $P$ is almost periodic. Otherwise, we only have $\cO(k^2)$ anchors for approximate circular occurrences (stemming from occurrences of some of $P_1$, $P_2$, or a repetitive region obtained by extending either of $P_1$ or $P_2$ in some direction).
\end{proof}

\section{Locked Fragments}
The notion of locked fragments originates from~\cite{DBLP:journals/siamcomp/ColeH02}. We use them as defined in \cite{FOCS20}. Let us state~\cite[Lemma 6.9]{FOCS20} \footnote{The original lemma also concluded that $L_1$ is a so-called $k$-locked prefix; however, this property is not needed here (and, in particular, a $k$-locked string is also locked).}
 with $d_S=112k$, for $k>0$; this characterization of locked fragments will be sufficient for our purposes. See \cref{July6} for an illustration.

\begin{lemma}[see {\cite[Lemmas 5.6 and 6.9]{FOCS20}}]\label{lem:locked_deomcp}
Let $S$ denote a string, $Q$ denote a primitive string, $q=|Q|$, and suppose that $\edl{S}{Q} \le 112k$
and $|S| \ge 225kq$
for some positive integer~$k$.

Then there is an algorithm which in $\Oh(k^2)$ time in the \pillar model computes disjoint \emph{locked fragments} $L_1,\ldots,L_\ell$ of $S$ satisfying:
\begin{enumerate}[(a)]
\smallskip
\item $S=L_1 Q^{\alpha_1} L_2 Q^{\alpha_2} \cdots L_{\ell-1} Q^{\alpha_{\ell-1}} L_\ell,$ where $\alpha_i \in \mathbb{Z}_{> 0}$ for all $i$,

\smallskip
\item $\edl{S}{Q} = \sum_{i=1}^\ell \edl{L_i}{Q}$ and $\edl{L_i}{Q}>0$ for all $i \in (1 \dd \ell)$,

\smallskip
\item $\ell=\Oh(k)$ and $\sum_{i=1}^\ell |L_i|  \le 676 kq$.
\end{enumerate}
\end{lemma}

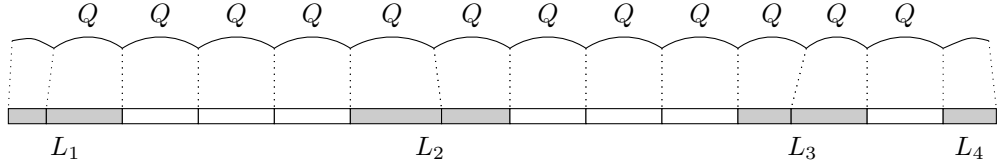
\begin{figure}[htpb]
\centering
\input{_fig_july6.tex}
\caption{Illustration of \cref{lem:locked_deomcp}. We have a decomposition $S=L_1\cdot Q^3 \cdot L_2 \cdot Q^3 \cdot L_3 \cdot Q^1 \cdot L_4$.  $L_1$ is an approximate suffix of $Q^{|S|}$, $L_4$ is an approximate prefix of $Q^\infty$, and internal gray parts are \emph{approximate} powers of $Q$. The remaining (white) fragments are \emph{exact} powers of $Q$.}\label{July6}
\end{figure}

Let us consider the decompositions obtained by applying \cref{lem:locked_deomcp} to strings $U$ and~$V$ from \EPSM w.r.t.\ the string $Q$.
Strings $U$ and $V$ are almost $Q$-periodic and almost $Q'$-periodic,
respectively, so $\edl{U}{Q},\edl{V}{Q} \le 112k$.
Moreover, $|U|,|V| \ge m \ge 256kq > 225k q$.
Thus, $U$ and $V$ satisfy the assumptions of the lemma.
If any of the decompositions starts with a locked prefix of length smaller than $q$ (possibly empty) or ends with a locked suffix of length smaller than $q$, we extend the locked fragment by a copy of~$Q$ and possibly by a neighbouring locked fragment if this copy was the only copy separating them. The total length of the locked fragments increases by at most $2q \le 2kq$, so it is bounded by $678kq$.

\section{Overlap Case of \EPSM}\label{sec:case1}
We consider all possible \emph{offsets} $\Delta$ (integers $\Delta \in (-|V| \dd |U|)$) by which we can shift $V$, looking for a length-$m$ substring of $V$ that approximately matches a substring of $U$.

We denote  $\Ext_t(X)=\bigcup_{x \in X} \{y\,:\, |x-y|\le t\}$.
Denote also by $\locked(U)$, $\locked(V)$ the set of  positions in all locked fragments in $U$, $V$, respectively.
\begin{definition}\label{def:ov_off}
$\Delta$ is a \emph{$t$-overlap offset} if there are positions $p, x$ such that $p-x =\Delta$, and
 \[p\in X\oplus \{-m,0,m\}
 ,\ \ 
x\in Y\oplus \{-m,0,m\}
\ \ \text{where } X=\Ext_t(\locked(U))\text{, } Y=\locked(V).\]
Otherwise $\Delta$ is a \emph{$t$-non-overlap offset}.
\end{definition}

An integer $\Delta$ is called a \emph{valid offset} if $\Delta\, \equiv_{77k}\, r \pmod{q}$. (Recall the definition of $r$ in \EPSM.) For two integer sets $A$ and $B$, let $A\ominus B = \{a - b\::\; a\in A,\, b \in B\}$.

\begin{observation}\label{obs:intervals}
For any intervals $I, J$, the set $\Ext_t(I\ominus J)$
is an interval of size $|I|+|J|-1+2t$ that can be computed in $\Oh(1)$ time.
\end{observation}
\begin{lemma}\label{lem:count_overlap}
The set of valid $t$-overlap offsets can be represented as a union of $\Oh(k^2+k^2t/q)$ intervals of length $\Oh(k)$ each. This representation can be computed in $\Oh(k^2+k^2t/q)$ time in the \pillar model.
\end{lemma}
\begin{proof}
Let $\ell_1,\ldots,\ell_{n_1}$ and $\ell'_1,\ldots,\ell'_{n_2}$ be the lengths of locked fragments in $U$ and $V$, respectively, and
$s_1=\sum_{i=1}^{n_1} \ell_i$, $s_2=\sum_{i=1}^{n_2} \ell'_i$.
By point (c) in \cref{lem:locked_deomcp}, we have $n_1+n_2=\Oh(k)$ and $s_1+s_2=\Oh(kq)$. 
By \cref{obs:intervals}, the set of $t$-overlap offsets is a union of $\Oh(k^2)$ intervals of total length proportional to:
\[\sum_{i=1}^{n_1} \sum_{j=1}^{n_2} (\ell_i+\ell'_j+t)=n_1n_2t+n_2\sum_{i=1}^{n_1} \ell_i+n_1\sum_{j=1}^{n_2}\ell'_j\le n_1n_2t+(n_1+n_2)(s_1+s_2)=\Oh(k^2(t+q)).\]
The intervals can be computed in $\Oh(k^2)$ time.
An interval of length $\ell$ contains $\Oh(k+\ell k/q)$ valid offsets  grouped into $\Oh(1+\ell/q)$ intervals of length $\Oh(k)$ each. These maximal intervals of offsets can be computed in $\Oh(1+\ell/q)$ time via elementary modular arithmetics. Therefore, the number of intervals of $t$-overlap offsets that are valid is proportional to
\[\big(\sum_{i=1}^{n_1} \sum_{j=1}^{n_2} 1\big) + \Oh(k^2+k^2t/q)=\Oh(k^2+k^2t/q)\]
and all of them can be computed in $\Oh(k^2+k^2t/q)$ time.
\end{proof}

\noindent
An app-match $(p,x)$ is called a $t$-{\bf overlap app-match} if and only if $p-x$ is a $t$-overlap offset.  In this section, we consider 
$t$-overlap app-matches. In \cref{sec:case234}, we consider $t$-non-overlap app-matches: app-matches $(p,x)$ such that $p-x$ is a $t$-non-overlap offset, for $t=\Theta(qk)$.

It follows from the statement of \EPSM that if $(p,x)$ is an app-match, then $p-x$ is a valid offset.
The following fact, together with \cref{lem:report-anchored}, implies a fast algorithm for computing the following set for a given offset $\Delta$:
\[\{p\in \Occ_k(V^{(x)},U)\,:\; \Delta=p-x,\,\Delta \equiv_{77k} r \pmod{q}\}.\]

\begin{fact}\label{June27}
If $(p,x)$ is an app-match, $\Delta=p-x$ and $\Delta'=m-\alpha+\Delta$, then the corresponding circular $k$-edit occurrence $U[p \dd p']$ is anchored at a position in $[\Delta' -k\dd \Delta'+k]$; see \cref{fig:anchor}.
\end{fact}

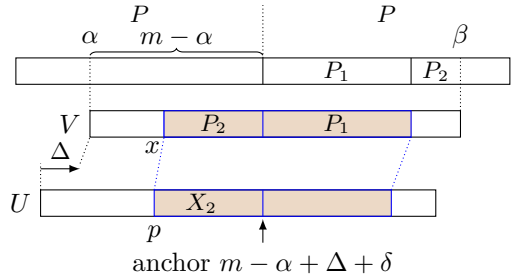
\begin{figure}[htpb]
\vspace*{-0.5cm}
\centering
\input{_fig_anchor.tex}
\caption{The anchor in $U$ is at position $m-\alpha+\Delta+\delta$, where $\delta=|X_2|-|P_2|\in[-k\dd k]$ (since $\ed(X_2,P_2) \le k$).}
\label{fig:anchor}
\end{figure}

Using \cref{lem:report-anchored,lem:decision-anchored} we obtain the following corollary.

\begin{corollary}\label{MaybeJuly6}
  Let $I$ be an interval of size $\Oh(k)$. All positions $p$ for which there exists an app-match $(p,x)$ such that $p-x\in I$,
  represented as a union of $\Oh(k^3)$ intervals, can be computed in $\Oh(k^3)$ time in the \pillar model.
  Moreover, one can check if there is any app-match $(p,x)$ with $p-x\in I$ in $\Oh(k^2 \log^3 k)$ time in the \pillar model.
\end{corollary}

The solution of the overlap case is presented in \cref{algo:ov}.
\cref{lem:count_overlap} together with \cref{June27,MaybeJuly6} imply the following lemma.

\begin{algorithm}[h!]

\vspace*{0.2cm}
Compute the decompositions of $U$ and $V$ into locked fragments\;

\medskip
// Compute the set $\Lambda$ of $(t+k)$-overlap offsets, being a union of $\Oh(k^2)$ intervals:\\
\ForEach{locked fragment $U[i_{\min} \dd i_{\max}]$}{
    \ForEach{locked fragment $V[j_{\min} \dd j_{\max}]$}{
        $\Lambda:=\Lambda \cup ([i_{\min}-j_{\max}-(t+k)\dd i_{\max}-j_{\min}+(t+k)] \oplus \{-m, 0, m\})$\;
    }
}

\medskip
// Compute the set $\Gamma$ of valid $(t+k)$-overlap offsets,\\
// represented as a union of $\Oh(k^2+k^2(t+k)/q)$ intervals of size $\Oh(k)$ each:\\
\ForEach{interval $I$ of offsets \KwSty{in} $\Lambda$}{
    $\Gamma:=\Gamma \cup \{\text{maximal intervals representing }\{i \in I\,:\,i \equiv_{77k} r \pmod{q}\}\}$\;
}

\medskip
\ForEach{interval $[i_{\min} \dd i_{\max}]$ of offsets \KwSty{in} $\Gamma$, with $i_{\max}-i_{\min}=\Oh(k)$}{
    $J:=[i_{\min} \dd i_{\max}]\oplus(m-\alpha)$\;
    \textbf{report} $\bigcup_{a \in J} \anchored_k(P,U,a)$\;
}
\caption{\large \bf Overlap case: reporting version}\label{algo1}\label{algo:ov}
\vspace*{0.2cm}
\end{algorithm}

\begin{lemma}\label{lem:case1}
Let $B$ be the output of \cref{algo:ov}. Then $B \subseteq \cycoc_k(P,U)$ and every $t$-overlap app-match  
occurrence $p$ is in $B$.

Moreover, if $t=\Oh(kq)$, \cref{algo:ov} works in $\Oh(k^{6})$ time in the \pillar model with the output represented as a union of $\Oh(k^{6})$ intervals.
\end{lemma}
\begin{proof}
Consider a $t$-overlap app-match $(p,x)$.
Then, there exists an anchor $a$ such that $p \in \anchored_k(P,U,a)$, and $y = a-(m-\alpha)$ is a $(t+k)$-overlap offset, since we have
\[\ed(V[x \dd m-\alpha),U[p \dd a)) + \ed(V[m-\alpha\dd x+m), U[a \dd p']) \leq k.\]
Now, $y$ is in some interval $[i_{\min} \dd i_{\max}] \in \Gamma$, as the union of the elements of $\Gamma$ comprises the set of valid $(t+k)$-overlap offsets.
Then, since $y \in [i_{\min} \dd i_{\max}]$, we have $a = y+(m-\alpha) \in [i_{\min} \dd i_{\max}]\oplus(m-\alpha)$, and hence $a$ is in one of the sets $J$ constructed in the penultimate line of \cref{algo:ov}.
In the case when $t=\cO(kq)$, using \cref{lem:count_overlap}, we compute, in $\cO(k^3)$ time, $\cO(k^3)$ intervals of anchors, of size $\cO(k)$ each.
The time complexity and the fact that the algorithm returns the output as a union of $\cO(k^6)$ intervals follows by a direct application of \cref{MaybeJuly6} to each interval of anchors.
\end{proof}

\noindent
To obtain the next corollary, we replace the last line of \cref{algo:ov} by:

\medskip
\textbf{if} $\anyanchored_k(P,U,J) \ne none$ \textbf{then return} $\anyanchored_k(P,U,J)$;

\begin{corollary}\label{corr:ov_dec}
If $t=\Oh(kq)$, one can check if $B \ne \emptyset$ and, if so, return an arbitrary element of $B$, in $\Oh(k^{5} \log^3 k)$ time in the \pillar model.
\end{corollary}

\section{Non-Overlap Case of \EPSM}\label{sec:case234}
Recall that an app-match $(p,x)$ is called a \emph{$t$-non-overlap app-match} if and only if
$p-x$ is a $t$-non-overlap offset.
In this section we assume $t=\Theta(kq)$.
The set of $t$-non-overlap offsets is too large, but it has a short representation.

\begin{lemma}\label{lem:nov-off} 
The set of $t$-non-overlap offsets can be partitioned into $\Oh(k^2)$ maximal intervals in $\Oh(k^2 \log \log k)$ time in the \pillar model.
\end{lemma}
\begin{proof}
There are $\Oh(k)$ locked fragments in $U$ and $V$.
By \cref{obs:intervals}, every pair of locked fragments, one from $U$ and one from $V$, induces an interval of $t$-overlap offsets that can be computed in $\Oh(1)$ time.
The complement of the union of these offsets can be computed in $\Oh(k^2 \log \log k)$ time by sorting the endpoints of the intervals using integer sorting~\cite{DBLP:journals/jal/Han04}.
\end{proof}
We denote by $\DD(t)$ the set of maximal intervals yielded by the above lemma. 
For simplicity, we mostly discuss the decision version of the problem in this section; the correctness proof for the reporting version requires a few further technical arguments.

Let $\lambda_k=(112k+3)\cdot(3k+10)\cdot q + 678kq$.

\begin{lemma}\label{lem:usemesomewhere}
If $\lambda_k>\frac{m}{2}$, \EPSM can be solved in $\Oh(k^5)$ time in the \pillar model, with the output represented as a union of $\Oh(k^5)$ intervals.
\end{lemma}
\begin{proof}
We have $m=\Oh(k^2q)$. As $\cO(k)$ out of every $q$ consecutive offsets are valid and they can be grouped in at most two intervals, there are $\cO(mk/q) = \Oh(k^3)$ valid offsets, which are grouped into $\Oh(k^2)$ intervals of size $\Oh(k)$ each. Let the set of such intervals be $\mathcal{J}$.
The time complexity and output size follow from an application of \cref{MaybeJuly6} to the $\cO(k)$-size interval of anchors corresponding to each $J \in \mathcal{J}$, as in the last three lines of \cref{algo:ov}.
\end{proof}

\noindent Henceforth we assume that $\lambda_k \le \frac{m}{2}$.\
Let $W$ be the longest fragment of $V$ such that each length-$m$ fragment of $V$ contains $W$, i.e., $W=V[|V|-m \dd m)$.

\begin{observation}
If $\lambda_k \le \frac{m}{2}$, then $W$ contains a fragment equal $Q^{3k+9}$ that is disjoint from locked fragments in $V$.
\end{observation}
\begin{proof}
We have $|W| \ge \frac{m}{2}$ since $|V| \le \frac32 m$.
By \cref{lem:locked_deomcp}, $V$ contains at most $112k+2$ locked fragments. Their total length does not exceed $678kq$.
By the pigeonhole principle, as $\lambda_k =(112k+3)\cdot((3k+10)\cdot q) + 678kq \le |W|$, string $W$ contains a substring of length at least $(3k+10)q$ that is disjoint from locked fragments.
By \cref{lem:locked_deomcp}, this substring is a substring of $Q^\infty$ and thus contains a copy of $Q^{3k+9}$.
\end{proof}

\begin{definition}[{\bf sample}]\label{def:sfr} 
We select an arbitrary fragment $V[j\dd j']$  of $W$ that equals $Q^{3k+9}$ and is disjoint from locked fragments in $V$. Then the middle fragment $V[j_1\dd j_2]$ of $V[j\dd j']$ that equals $Q^{k+1}$ becomes an additional locked fragment. The fragment $V[j_1\dd j_2]$ is called the \emph{sample}.
\end{definition}

When computing $t$-overlap offsets with the algorithm of \cref{sec:case1}, we treat the sample as a locked fragment; the total length of the locked fragments is then still $\Oh(kq)$.

Henceforth we replace $P$ by its rotation $\rot^{y}(P)$, where $y=(j_1+\alpha) \bmod m$. Let us note that after this change, the sets $\anchored_k$ can be computed equally efficiently as the sets $\anchored_k$ for the original $P$. This follows from the fact that the algorithm underlying \cref{lem:report-anchored} does not use $\IPM$ queries, and the remaining queries from the \pillar model can easily be implemented in $\Oh(1)$ time if an input string is given by its cyclic rotation.

For an interval $I=[i_1\dd i_2]$ and a string $S$, by $S[I]$ we denote $S[i_1 \dd i_2]$. 
We denote $\q=2(k+6)(q+3)$; the constants originate from the proof of \cref{lem:correctness}.

\newcommand{\C}{\mathbf C}
\begin{observation}\label{obs:special_locked_disj} Let $[d_1 \dd d_2]\in \DD(t)$. 
If $V[j_1\dd j_2]$ is the sample in $V$, then $U[j_1+d_1-t \dd j_2+d_2+t]$ does not contain a position in a locked fragment, since we defined the sample as an (exceptional) locked fragment.
\end{observation}

\begin{definition} For an interval $\D=[d_1\dd d_2]$,  denote  \[\scope(\D)=\Ext_k([j_1\dd j_2]\oplus \D),
\ \ \CritPos(\D)\,=\, \Occ_0(Q^{k+1},\,  U[\scope(\D)]).\]
The positions in $\bigcup_{\,\D \in \DD(\q)}\, \CritPos(\D)$
are called \emph{critical positions}; see \cref{fig:przerobmnie}.
\end{definition}

The main idea of the proof of the next lemma is as follows: in an app-match for an offset from $\D$, at least one copy of $Q$ from the sample must match a copy of $Q$ in $\scope(\D)$ exactly. For $\D \in \DD(\q)$, $\scope(\D)$ is a substring of $Q^\infty$. This implies that the whole sample matches a fragment of $\scope(\D)$ exactly, which is how critical positions were defined.

We prove \cref{lem:critical} with the aid of the following well-known fact.

\begin{fact}\label{fct:remove}
For any two strings $A$ and $B$ and letter $c$, we have $\ed(Ac,Bc)=\ed(A,B)=\ed(cA,cB)$.
\end{fact}

\begin{lemma}\label{lem:critical}
For each position $p$ for which there is a $\q$-non-overlap app-match $(p,x)$, we have
$ p \in \bigcup\, \{\, \anchored_k(P,U,i)\,:\,
 i\ \mbox{is a critical position} \,\}$.
\end{lemma}
\begin{proof}
    Consider a $\hat{q}$-non-overlap app-match $(p,x)$, where $p-x$ belongs to an interval $\D = [d_1 \dd d_2] \in \DD(\q)$.
    
    Due to \cref{obs:special_locked_disj}, no position in $U[I]$, where $I = [j_1+d_1 - \q \dd j_2+d_2 + \q]$, belongs to a locked fragment.
    Note that $\scope(\D) = [j_1+d_1-k \dd j_2+d_2+k] \subseteq I$.
    Hence, $U[\scope(\D)]$ is a fragment of $U$ that is disjoint from all locked fragments and is thus equal to a substring of $Q^\infty$.

    Let $U[p \dd p')$ be a fragment of $U$ that is at edit distance at most $k$ from $V^{(x)}$.
    Further, let us fix an alignment of $V^{(x)}$ and $U[p \dd p')$ of cost $\ed(V^{(x)}, U[p \dd p'))$.
    
    Since $V[j_1 \dd j_2] = Q^{k+1}$, this alignment aligns at least one of the $k+1$ copies of $Q$ in the sample exactly with a copy of $Q$ in $U$.
    
    Let this copy of $Q$ be $V[j' \dd j'+q)$, where $j'=j_1+t\cdot q,\, t \in [0 \dd k]$, and suppose that it is aligned exactly with $U[z' \dd z'+q)$. We then have 
    \[\ed(V^{(x)}, U[p \dd p'))\; =\; \ed(V[x \dd  j'), U[p \dd z')) + 
                                   \ed(V[j'+q \dd  x+m), U[z' + q \dd p')).
    \]

\noindent Let $a = z' - t\cdot q$ and $z = z'+(k-t) \cdot q$. We observe that 
$[a \dd z) \subseteq \scope(\D)$, since \[z'\ \in\ [-k \dd k] \oplus ((p - x) + j')\ \subseteq\ [j'+d_1 -k \dd j'+d_2+k],\] 
which implies $a \in [j_1+d_1 -k \dd j_1+d_2+k]$ and $z \in [j_2+d_1 -k \dd j_2+d_2+k]$.
This means that $U[a \dd z)=Q^{k+1}$, so $a \in \CritPos(\D)$.

Repeatedly using \cref{fct:remove} for the first and last summands of the first summation below, we have:
\begin{multline*}
    \ed(V[x \dd j_1),U[p \dd a))\; +\; \ed(V[j_1 \dd j_2], U[a \dd z))\; + \;
                                 \ed(V(j_2 \dd x+m), U[z \dd p'))\;= \\
     \vspace*{2mm}\ed(V[x \dd j'),U[p \dd z')) + \ed(V[j'+q \dd x+m), U[z'+q \dd p'))
     =  \ed(V^{(x)}, U[p \dd p')),
    \end{multline*}
    since $\ed(V[j_1 \dd j_2], U[a \dd z))=0$.
    Position $a$ is a critical position and this concludes the proof of this lemma.
\end{proof}

The lemma says that it would be enough to consider 
$\anchored_k(P,T,i)$ for all critical positions $i$. Unfortunately,  the total number
of critical positions can be too large; however, they are grouped into $\Oh(k^2)$ arithmetic progressions and it is enough to consider 
the first and the last position in each such progression. 

In the decision version we use \cref{algo2dec}. A proof of the following \cref{lem:correctness} with several auxiliary lemmas is presented in \cref{sec:decision}. A generalization of \cref{algo2dec} to the reporting variant of the problem is presented in \cref{sec:algo2rep}. The subsections end with proofs of the decision and reporting version of \cref{thm:main_pillar}, respectively.

\begin{algorithm}[h!]

Compute decompositions of $U$ and $V$ into locked fragments and the sample\;

\vspace*{0.1cm}
Compute $\DD(\q)$\;

\vspace*{0.1cm}
\ForEach{interval of non-overlap offsets $\D \in \DD(\q)$}
{

\smallskip
$i_1 := \min\,  \CritPos(\D)$;\ \ $i_2 := \max\,  \CritPos(\D)$\;
 
\smallskip
\lIf{$\anyanchored_k(P,U,i_1) \ne none$}{\Return{$\anyanchored_k(P,U,i_1)$}}\smallskip
\lIf{$\anyanchored_k(P,U,i_2) \ne none$}{\Return{$\anyanchored_k(P,U,i_2)$}}
}
\smallskip
{\Return{$none$};
}
\caption{\bf Non-overlap case:  decision version}\label{algo2dec}
\vspace*{0.1cm}
\end{algorithm}

\begin{figure}[t]
\centering
\resizebox{\textwidth}{!}{\input{_fig_algorithm}}
\vspace*{-0.5cm}
\caption{Illustration of basic parameters in the algorithm: $\D=[d_1\dd d_2],\ I = \Ext_k([j_1\dd j_2]\oplus \D)$.  We have $I\cap \locked(U)=\emptyset$.
$\CritPos(\D)$ consists of  critical positions 
shown as green circles.}\label{fig:przerobmnie}
\end{figure}
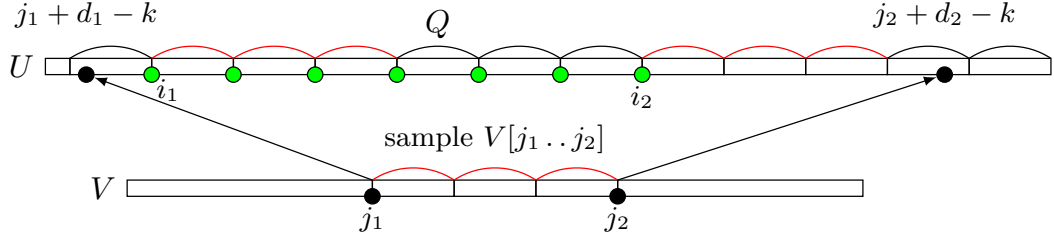

\begin{lemma}\label{lem:correctness}
Assume that $\lambda_k \le \frac{m}{2}$.
\cref{algo2dec} works in $\Oh(k^4)$ time in the \pillar model and returns a circular $k$-edit occurrence of $P$ in $U$ if any $\q$-non-overlap app-match exists.
\end{lemma}

\subsection{Proof of \texorpdfstring{\cref{lem:correctness}}{Lemma 36}} \label{sec:decision}
For a fragment $F=U[I]$ ($F=V[I]$, respectively), we denote by $\locked(F)$ the set $I\cap \locked(U)$ ($I\cap \locked(V)$,
respectively).
\begin{definition}
Two fragments $F_1, F_2$ (both of $U$ or both of $V$) are called \emph{locked-equivalent} if $\locked(F_1)=\locked(F_2)$ and there are no locked positions in a prefix and a suffix of length $(k+4)q$ in $F_1$ and in $F_2$; see~\cref{fig:locked_eq}.
\end{definition}

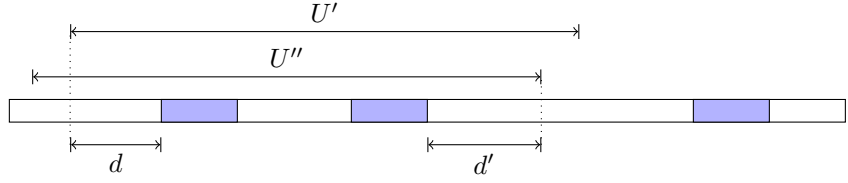
\begin{figure}[htpb]
\centering
\input{_fig_locked_eq.tex}
\caption{The gray boxes correspond to locked fragments, while $d,d'\ge (k+4)q$. The fragments $U'$ and $U''$ are locked-equivalent.  }\label{fig:locked_eq}
\end{figure}

We extend \cref{def:anchored} and say that a circular $k$-edit occurrence $T[p\dd p']$ of $P$ is \emph{$x$-anchored} at position $i$ if 
$\ed(T[p\dd i), P[x\dd m)) + \ed(T[i\dd p'], P[0\dd x)) \leq k.$
For a fragment $Y=X[i\dd j]$ and integer $y$, we denote $\shift(Y,y)=X[i+y\dd j+y]$.

By $\Frag(S)$ we denote the set of all fragments of string $S$.
For a fragment $X[i\dd j]$, we denote $\First(X[i \dd j])=i$. 
We recall that $\q=2(k+6)(q+3)$.
The following lemma states simple properties of $\q$-non-overlap offsets.

\begin{lemma}\label{lem:after_locked}
    Let $W \in \Frag(U)$ and $Z \in \Frag(V)$ and assume that $\ed(W,Z) \le k$ and $\First(W)-\First(Z) \in \DD(t)$. Assume that in an optimal alignment between $W$ and $Z$, position $i$ in $W$ is aligned with position $j$ in $Z$.
    
    If position $i$ in $W$ is in a locked fragment (from $U$) and $i \in [t \dd |W|-t)$, then positions $[j-t+k \dd j+t-k]$ in $Z$ are not in a locked fragment (from $V$).

    If position $j$ in $Z$ is in a locked fragment (from $V$) and $j \in [t \dd |Z|-t)$, then positions $[i-t+k\dd i+t-k]$ in $W$ are not in a locked fragment (from $U$).
\end{lemma}
\begin{proof}
Let $\Delta = \First(W)-\First(Z)$ and consider the first statement.
By the definition of $t$-non-overlap offsets, we have that all positions $x$ in $V$ such that $\First(W)+i - x \in [\Delta - t \dd \Delta +t]$ are disjoint from locked fragments.
These correspond to the positions $y$ of $Z$ that satisfy
\[\First(W)+i - (\First(Z)+y) \in [\Delta - t \dd \Delta +t] \Longleftrightarrow i-y \in [-t \dd t] \Longleftrightarrow y \in [i-t \dd i+t].\]
Now, since $j-i \in [-k \dd k]$, we have $i \in [j-k \dd j+k]$ and hence \[[j-t+k \dd j+t-k] \subseteq [i-t \dd i+t],\]
concluding the proof of the first statement.
The proof of the second statement is analogous.
\end{proof}

The next lemma heavily exploits properties of $\Theta(kq)$-non-overlap offsets.

\begin{lemma}\label{lem:greedy}
Consider a pair $X$ and $Y$ of fragments such that either $X\in \Frag(U)$ and $Y\in \Frag(V)$ or $X\in \Frag(V)$ and $Y\in \Frag(U)$.
Let 
$a=\First(X)$, $b=\First(X')$, $c=\First(Y)$, $X'=\shift(X,q)$, and $\Delta=1$ if $X \in \Frag(U)$, $\Delta=-1$ otherwise.
Suppose that $X,X'$ are locked-equivalent and that 
$(a-c)\cdot\Delta,(b-c)\cdot\Delta\in \DD((k+3)(q+1)).$

If $X=_k Y$ or $X'=_k Y$, then $\ed(X,Y)=\ed(X',Y)$.
\end{lemma}
\begin{proof}
    We assume that $X \in \Frag(U)$ and $Y \in \Frag(V)$, i.e., that $\Delta=1$; the opposite case is analogous.
    Let us focus on the case when $X=_k Y$ and $\ed(X,Y) \leq \ed(X',Y)$; the other case is symmetric.
    Let us order the locked fragments in both $X$ and $Y$ with respect to their starting positions in those strings. We call a locked fragment a \emph{breakpoint locked fragment} if the subsequent locked fragment in the defined order originates from a different string.
    We next show that there exists a sequence $(X_0,Y_0), \cdots, (X_t,Y_t)$ of pairs of strings that satisfies the following:
    \begin{itemize}
    \item $X=X_0\cdots X_t$ and $Y=Y_0\cdots Y_t$,
    \item $\ed(X,Y) = \sum_{i=0}^t \ed(X_i,Y_i)$,
    \item $X_i = Y_i = Q$ for odd $i$,
    \item for each pair $(X_i,Y_i)$, only one of $X_i$ or $Y_i$ contains locked fragments (naturally inherited from $X$ and $Y$).
    \end{itemize}
    Let us greedily construct this sequence of pairs given an optimal alignment between $X$ and~$Y$.
    
    Recall that $X$ has a prefix of length $(k+4)q$ that contains no locked positions.
    At least one of the first $k+1$ implied copies of $Q$ must be aligned exactly with a copy of $Q$ in $Y$; we set the first such exactly aligned copies to be $X_1$ and $Y_1$.
    
    Then, we repeatedly consider the subsequent $j$-th breakpoint locked fragment (where~$j$  starts from $1$) $L=Z[z_1 \dd z_2)$, where $\{Z,W\} = \{X,Y\}$, if one exists.
    Since $(a-c) \in \DD((k+3)(q+1))$, \cref{lem:after_locked} implies
    $S_L := Z[z_2 \dd z_2 + (k+1)q) = Q^{k+1}.$
    
    Then, at least one of the $k+1$ copies of $Q$ in $S_L$ must be aligned exactly with a copy of~$Q$ in~$W$; we set these copies to be $X_{1+2j}$ and $Y_{1+2j}$.
    When there are no further breakpoints to be considered, we simply consider a final odd pair of copies of $Q$ that are aligned exactly, such that one of them is the last copy of $Q$ in $X$ that is aligned exactly---recall that $X$ has a suffix of length $(k+4)q$ that contains no locked positions.
    Finally, we ensure that $t$ is even by appending a pair of empty strings if necessary.

    We have $X=X_0QX_2Q \cdots X_{t-2}QX_t$ where $X_0$ and $X_t$ are substrings of $Q^\infty$, so $X'= X_0X_2QX_4Q\cdots QX_{t-2} Q^2X_t$. Therefore
    \begin{align*}
    \ed(X',Y) & \leq \ed(X_0, Y_0) + \ed(X_2Q, QY_2) + \ldots + \ed(X_{t-2}Q, QY_{t-2}) + \ed(QX_t, QY_t)\\
            & \stackrel{(\star)}{=} \ed(Y_0, X_0) + \ed(X_2, Y_2)  + \ldots + \ed(X_t, Y_t) \\
            &= \ed(X,Y),
    \end{align*}
    where equality $(\star)$ follows from the fact that, for each pair $(X_i,Y_i)$, one of $X_i$ and $Y_i$ is a power of $Q$, which allows us to apply \cref{fct:remove} to remove a pair of copies of $Q$ from $\ed(QX_t, QY_t)$ and from each pair of the form $(X_{2d}Q,QY_{2d})$. This concludes the proof of the lemma since we have assumed that $\ed(X,Y) \leq \ed(X',Y)$.
\end{proof}

Let $i_1 = \min\,  \CritPos(\D)$, $i_2 = \max\,  \CritPos(\D)$ as in \cref{algo2dec}.
The next lemma shows that in many cases, if $U[p \dd p']$ forms a $\q$-non-overlap app-match that is anchored at a critical position $i$ such that $i_1 < i < i_2$, then the same fragment or a fragment shifted by $q$ positions forms a $\q$-non-overlap app-match anchored at a critical position $i \pm q$.

We refer to \cref{def:anchored} for the meaning of 
\emph{$x$-anchored}.

\begin{lemma}\label{lem: non-overlap pointwise}
Let $V[j_1 \dd j_2]=Q^{k+1}$ be the sample, $\C=\CritPos(\D)$ where $\D \in \DD(\q/2)$, and $i \in \C$.
If $I=[p \dd p']$ and $U[I]$ is $x$-anchored at $i$, then for any $y\in \{q,-q\}$:
\vspace*{0.1cm}
\begin{enumerate}[(a)]
    \item\label{itA} If $U[I]$ and $U[I\oplus y]$ are locked-equivalent and $i+y\in \C$, then $U[I\oplus y]$ is $x$-anchored at $i+y$. \vspace*{0.1cm}
    \item\label{itB} If $V'=V^{(x)}$ and $\shift(V',y)$ are locked-equivalent and $i-y \in \C$, then $U[I]$ is $(x+y)$-anchored at $i-y$.
\end{enumerate}
\end{lemma}
\begin{proof}
We give separate, though to a large extent similar, proofs of both points.

\proofsubparagraph{Point \eqref{itA}:} 
First let $y=-q$.
As $U'$ and $\shift(U',-q)$ are locked-equivalent, $U[p-q\dd p)$ and $U(p'-q\dd p']$ do not overlap any locked fragment in $U$.

We know that $i \in [j_1+d_1-k\dd j_1+d_2+k]$, where $\D=[d_1 \dd d_2] \in \DD(\q/2)$. 
Hence, $i-j_1 \in [d_1-k \dd d_2+k]$ is a $(k+3)(q+2)$-non-overlap offset. Therefore, $p-x$ is a $(k+3)(q+1)$-non-overlap offset, as $U[p\dd i)=_kV[x\dd j_1)$. Moreover, because $i-q-j_1$ is a $(k+3)(q+2)$-non-overlap offset, $p-q-x$ is a $(k+3)(q+1)$-non-overlap offset (as $|(p-x) - (i-j_1)| \le k$).

We have the following properties:
\begin{itemize}
    \item By \cref{lem:after_locked}, $U[p \dd i)$ ends with $Q^{k+5}$, so $U[p-q \dd i-q]$ ends with $Q^{k+4}$.
    \item Hence, $U[p \dd i)$ and $U[p-q \dd i-q)$ are locked-equivalent, as $U[p \dd p']$ and $U[p-q \dd p'-q]$ are locked-equivalent.
    \item Similarly we obtain that $U[i \dd p')$ and $U[i-q \dd p'-q)$ are locked-equivalent.
\end{itemize}

Consequently, we can apply \cref{lem:greedy} (with $\Delta=1$) to obtain the following:
\begin{align*}
\ed(U[p\dd i),V[x\dd j_1))&=\ed(U[p-q\dd i-q),V[x\dd j_1))\text{ and}\\
\ed(U[i\dd p'),V[j_1\dd x+m))&=\ed(U[i-q\dd p'-q),V[j_1\dd x+m)).
\end{align*}
Thus $U[p-q\dd p'-q]$ is $x$-anchored at $i-q$.
The proof that $U[p+q\dd p'+q]$ is $x$-anchored at $i+q$ is symmetric.

\proofsubparagraph{Point \eqref{itB}:} 
Let us again start with the case $y=-q$.

Let us denote $x'=x+m-1$.
By the assumption, $V[x-q\dd x)$ and $V(x'-q\dd x']$ do not overlap any locked fragment in $V$.

We know that $i-j_1$ is a $(k+3)(q+2)$-non-overlap offset. Hence, $p-x$ is a $(k+3)(q+1)$-non-overlap offset, as $U[p\dd i)=_k V[x\dd j_1)$. Moreover, because $i+q-j_1$ is a $(k+3)(q+2)$-non-overlap offset, $p+q-x$ is a $(k+3)(q+1)$-non-overlap offset.

We have the following properties:
\begin{itemize}
    \item Substring $V[x-q \dd j_1)$ ends with $Q^{k+4}$ and $V[x \dd j_1+q)$ ends with $Q^{k+5}$.
    \item Hence, $V[x \dd j_1+q)$ and $V[x-q \dd j_1)$ are locked-equivalent, as $V^{(x)}$ and $V^{(x-q)}$ are locked equivalent.
    \item Similarly we obtain that $V[j_1+q \dd x']$ and $V[j_1 \dd x'-q]$ are locked-equivalent.
\end{itemize}

Consequently, we can apply \cref{lem:greedy} (case $\Delta=-1$) to obtain the following. In the first equalities, we extend (shorten, respectively) the two substrings by a suffix (prefix, respectively) that is copy of $Q$ (cf.\ \cref{fct:remove}). Moreover, this does not change the edit distance of the two substrings in scope. This operation is possible because the length of the suffixes to be shortened, i.e., $p'-i$ and $x+m-j_1$, are at least $q$.
Indeed, we have $x+m-1-q \ge j_2+q(k+1)$ (the extended sample $Q^{3k+9}$ is necessarily a fragment of $V^{(x-q)}$) and $j_2 \ge j_1$, so $x+m-1-q \ge j_1 + k(q+1)$ and $x+m-j_1>q+k(q+1)$. Further, $p'-i \ge x+m-1-j_1-k>q+kq$.
\begin{align*}
\ed(U[p\dd i),V[x\dd j_1))&=\ed(U[p\dd i+q),V[x\dd j_1+q))\\
&=\ed(U[p\dd i+q),V[x-q\dd j_1))\text{ and}\\
\ed(U[i\dd p'),V[j_1\dd x'])&=\ed(U[i+q\dd p'),V[j_1+q\dd x'])\\
&=\ed(U[i+q\dd p'),V[j_1\dd x'-q]).
\end{align*}
Thus $U[p\dd p']$ is $(x-q)$-anchored at $i+q$.
The proof that $U[p\dd p']$ is $(x+q)$-anchored at $i-q$ is symmetric.
\end{proof}

\medskip
The sets $\anchored_k$ contain too little information for proving the correctness of the algorithm. It is important that for any of the $\Oh(k^2)$ intervals of  positions of app-matches $[p_l \dd p_r]$ returned by a call to $\anchored_k(P,U,i)$, there exist positions $[p'_l \dd p'_r]$ and values $[x_l \dd x_r]$ of cyclic rotations such that $U[p_l \dd p'_l]$ is $x_l$-anchored at $i$, $U[p_l+1 \dd p'_l+1]$ is $(x_l+1)$-anchored at $i$, etc.
Therefore we define
\[\anchored'_k(P,T,i)\,=\, \{\,(p,p',x)\,:\, T[p\dd p'] \ \mbox{is }x\mbox{-anchored at}\ i\,\}.\]
The notation lets us restate \cref{lem: non-overlap pointwise} equivalently as follows.

\begin{lemma}[Equivalent statement of \cref{lem: non-overlap pointwise}]\label{lem: non-overlap pointwise2}~\\
Let $V[j_1 \dd j_2]=Q^{k+1}$ be the sample and $C=\CritPos(\D)$ where $\D \in \DD(\q/2)$.
Assume that $(p,p',x) \in \anchored'_k(P,U,i)$ for some $i \in C$. For $y\in \{q,-q\}$ we have
\vspace*{0.2cm}
\begin{enumerate}[(a)]
    \item\label{itA2} If $U'=U[p \dd p']$ and $\shift(U',y)$ are locked-equivalent and $i+y\in C$,  then \[(p+y,p'+y,x) \in \anchored'_k(P,U,i+y).\]
    \item\label{itB2} If $V'=V^{(x)}$ and $\shift(V',y)$ are locked-equivalent and $i-y \in C$,
    then
\[(p,p',x+y)\in \anchored'_k(P,U,i-y).\]
\end{enumerate}
\end{lemma}

For a triad $(I,J,L)$ of intervals of the same size, we denote the combined set of triples
\[\Triples(I,J,L)\;=\;\{(a+t, b+t,c+t)\,:\, 0\le t<|I|\},\ \mbox{where}\ (a,b,c)=
(\min(I),\,\min(J),\,\min(L)). \]
For example $\Triples([1 \dd 3],[5 \dd 7],[2 \dd 4])=\{(1,5,2),(2,6,3),(3,7,4)\}$.
(Treating $I,J,L$ as lists, this can be written in Python as set(zip($I,J,L$)).\,)
Just like \cref{lem: non-overlap pointwise} states a relation of single elements of the sets $\anchored_k$ for anchors at two consecutive critical positions, the next lemma shows what happens to intervals of positions in $\anchored_k$ (together with end-positions of app-matches and the rotations of $P$).

Denote by $\lcut_q(I),\rcut_q(I)$ the operations of removing from the interval $I$
its prefix/suffix of length $q$, possibly obtaining an empty interval. For example, $\lcut_2([2\dd 5])=[4\dd 5]$.

For every $p \in \anchored_k(P,U,i)$ that satisfies the assumption of \cref{lem: non-overlap pointwise}\eqref{itB} and $i_1 < i < i_2$, that lemma immediately shows that $p \in \anchored_k(P,U,i-q) \cap \anchored_k(P,U,i+q)$. Unfortunately, this 
assumption does not always hold. However, the following \cref{lem:left_right} shows that this is true for all but at most $q$ elements $p \in \anchored_k(P,U,i)$.

To prove \cref{lem:left_right}, roughly speaking, we compute a \emph{superposable partition} of intervals $I_1,I_2,I_3,I_3\oplus m,$
such that in each part, locked fragments can occur only in the parts originating from one of the strings $U$, $V$. As before, this is possible thanks to the fact that the offset is non-overlapping; here we use the fact that the definition of $t$-non-overlap offsets (\cref{def:ov_off}) covers the cases $(\Delta' \pm m) \oplus [-t \dd t]$. Finally, we apply the appropriate point of \cref{lem: non-overlap pointwise} to positions in each part in bulk.

\begin{lemma}\label{lem:left_right}
Let $\D \in \DD(\q), i_1=\min \CritPos(\D)$, $i_2=\max \CritPos(\D)$. Assume that for some $i \in \CritPos(\D)$ such that $i \ne i_1,i_2$, we have
$\Triples(I_1,I_2,I_3)\subseteq \anchored'_k(P,U,i)$,
where $|I_1|=|I_2|=|I_3|\ge q$.
Then:
\begin{align*}
\Triples(\lcut_q(I_1),\lcut_q(I_2),\rcut_q(I_3)) &\subseteq \anchored'_k(P,U,i+q),\\
\Triples(\rcut_q(I_1),\rcut_q(I_2),\lcut_q(I_3) &\subseteq \anchored'_k(P,U,i-q).
\end{align*}
\end{lemma}
\begin{proof}
Let us denote $x'_l=x_l+m-1$ and $x'_r=x_r+m-1$.
We select indices \[p_1,\ldots,p_d=p_r,\, p'_1,\ldots,p'_d=p'_r,\, x_1,\ldots,x_d=x_r,\,
x'_1,\ldots,x'_d,\ \mbox{where}\ x'_i=x_i+m-1\]
and sentinel indices 
$p_0=p_l+q,\, p'_0=p'_l+q,\, x_0=x_l+q,\, x'_0=x'_l+q$ such that:
\begin{enumerate}[(1)]
\item The indices are equally spaced within the intervals: $p_a-p_l=p'_a-p'_l=x_a-x_l$ for each $a \in [1 \dd d]$.
\item For each $a \in [1 \dd d-1]$, the substrings $U[p_a-q \dd p_a)$, $U[p'_a-q \dd p'_a)$, $V[x_a-q \dd x_a)$, $V[x'_a-q \dd x'_a)$ contain no positions from locked fragments.
\item\label{it:lock_equiv} For each $a \in [1 \dd d]$:
\begin{itemize}
    \item either none of the substrings $U[p_{a-1}-q \dd p_a]$, $U[p'_{a-1}-q \dd p'_a]$ contains a position from a locked fragment and each of the substrings $V[x_{a-1}-q \dd x_a]$, $V[x'_{a-1}-q \dd x'_a]$ contains no locked positions at its prefix and suffix of length $q(k+4)$,
    \item or none of the substrings $V[x_{a-1}-q \dd x_a]$, $V[x'_{a-1}-q \dd x'_a]$ contains a position from a locked fragment and each of the substrings $U[p_{a-1}-q \dd p_a]$, $U[p'_{a-1}-q \dd p'_a]$ contains no locked positions at its prefix and suffix of length $q(k+4)$.
\end{itemize}
\end{enumerate}

Such indices can always be selected thanks to the fact that $i-j_1 \in \DD(\q)$ (cf.\ \cref{lem:after_locked}).
Let us consider each $a \in [1 \dd d]$. Assume first that none of the substrings $U[p_{a-1}-q \dd p_a]$, $U[p'_{a-1}-q \dd p'_a]$ contains a position from a locked fragment. We know that 
\[\Triples([p_{a-1}-q\dd p_a-q],[p'_{a-1}-q\dd p'_a-q],[x_{a-1}-q\dd x_a-q]) \subseteq \anchored'_k(P,U,i).\]
\noindent Hence, by \cref{lem: non-overlap pointwise2}\eqref{itA2},
\begin{equation}\label{eq1}
\Triples([p_{a-1}\dd p_a],[p'_{a-1}\dd p'_a],[x_{a-1}-q\dd x_a-q]) \subseteq \anchored'_k(P,U,i+q).
\end{equation}
Let us note that \cref{lem: non-overlap pointwise2} can be applied (here and below in the proof) thanks to long fragments without locked positions that are guaranteed by point \eqref{it:lock_equiv}.

We also know that
$\Triples([p_{a-1}\dd p_a],[p'_{a-1}\dd p'_a],[x_{a-1}\dd x_a]) \subseteq \anchored'_k(P,U,i).$
\noindent Hence, by \cref{lem: non-overlap pointwise2}\eqref{itA2},
\begin{equation}\label{eq2}
\Triples([p_{a-1}-q\dd p_a-q],[p'_{a-1}-q\dd p'_a-q],[x_{a-1}\dd x_a]) \subseteq \anchored'_k(P,U,i-q).
\end{equation}

Assume now that none of fragments  $V[x_{a-1}-q \dd x_a]$, $V[x'_{a-1}-q \dd x'_a]$ contains a locked position. We know that 
$\Triples([p_{a-1}\dd p_a],[p'_{a-1}\dd p'_a],[x_{a-1}\dd x_a]) \subseteq \anchored'_k(P,U,i).$
\noindent By \cref{lem: non-overlap pointwise2}\eqref{itB2}, we obtain \eqref{eq1}.
We also know that
\[\Triples([p_{a-1}-q\dd p_a-q],[p'_{a-1}-q\dd p'_a-q],[x_{a-1}-q\dd x_a-q]) \subseteq \anchored'_k(P,U,i).\]
By \cref{lem: non-overlap pointwise2}\eqref{itB2}, we obtain \eqref{eq2}.
Taking a union over all $a \in [1 \dd d]$, we obtain the conclusion since $p_0$ was defined as $p_l+q$, and $p'_0$ as $p'_l+q$ and $x_0$ as $x_l+q$.
\end{proof}

We show that $\anchored'_k(P,T,i)$ can be represented by a set $\R_k(P,U,i)$ of triads (the set consists of $\Oh(k^2)$ triads) in the sense that 
\[\anchored'_k(P,T,i)\ =\, \bigcup\,\{\,\Triples(I,J,L)\,:\, (I,J,L)\in \R_k(P,T,i)\}.\]
The triad notation will be important in the proof of \cref{lem:correctness}.
The next fact readily follows from the construction of the set $\anchored_k$ (the proof of \cref{{lem:report-anchored}} presented in \cite{ESA22}).

\begin{fact}\label{fct:max}
For any index $i$, the set $\anchored'_k(P,U,i)$ is represented as a collection $\R_k(P,U,i)$ of triads, such that
for any triad $([p_l\dd p_r],[p'_l\dd p'_r],[x_l\dd x_r])$, we have:
\vspace*{1mm}\noindent
\begin{itemize}
\item $p_l=0$ or $x_l=0$ or $U[p_l-1] \ne V[x_l-1]$
\item $p_r'=|U|-1$ or $x_r+m=|V|$ or $U[p_r'+1] \ne V[x_r+m]$.
\end{itemize}
\vspace*{1mm}\noindent
In particular, if $p_l,x_l>0$, then $p_l-1$ is in $\locked(U)$ or $x_l-1$ is in $\locked(V)$. Symmetrically, if $p'_r+1<|U|$ and $x_r+m<|V|$, then $p'_r+1$ is in $\locked(U)$ or $x_r+m$ is in $\locked(V)$. 
\end{fact}

We say that a triad $([p_l\dd p_r],[p'_l\dd p'_r],[x_l\dd x_r]) \in \R_k(P,U,i)$ is \emph{left-$U$-locked} if $p_l=0$ or any of the positions $p_l-1$, $p'_l-1$ is in a locked fragment. Let us note that in the case that $p_l=0$, $p_l \in \locked(U)$ by definition.

Similarly, a triad $([p_l\dd p_r],[p'_l\dd p'_r],[x_l\dd x_r]) \in \R_k(P,U,i)$ is \emph{right-$U$-locked} if $p'_r=|U|-1$ or any of the positions $p_r+1$, $p'_r+1$ is in a locked fragment.

Symmetrically, a triad is called left-$V$-locked if $x_l=0$ or any of the positions $x_l-1$, $x'_l-1$ is in a locked fragment and right-$V$-locked if $x'_r=|V|-1$ or any of the positions $x_r+1$, $x'_r+1$ is in a locked fragment.

By \cref{lem:left_right}, if $I \subseteq \anchored_k(P,U,i)$ for an interval $I$, then we have
$\lcut_q(I) \subseteq\anchored_k(P,U,i+q)$ and $\rcut_q(I)  \subseteq \anchored_k(P,U,i-q)$.
In the proof of \cref{lem:correctness}, we use \cref{lem: non-overlap pointwise} on positions in the first and last $q$ positions of $I$ to show that one of the following conditions hold:
\[(\star)\ I \subseteq \anchored_k(P,U,i\pm q)\ 
\text{or} \ (\star\star)\ I \ominus q \subseteq \anchored_k(P,U,i-q).\]
In case $(\star)$, by induction we show that $I \subseteq \anchored_k(P,U,i_1) \cup \anchored_k(P,U,i_2)$. In case $(\star\star)$, we show by induction that $J:=I \oplus (i_1-i) \subseteq \anchored_k(P,U,i_1)$. 
We are now ready to prove correctness of \cref{algo2dec}.

\begin{proof}[Proof of \cref{lem:correctness}]
The complexity of \cref{algo2dec} directly follows from \cref{lem:report-anchored} (computing $\anchored_k$), \cref{lem:locked_deomcp} (computing decompositions into locked fragments) and \cref{lem:nov-off} (computing $\DD(\q)$). 

By \cref{lem:critical}, for a given interval of offsets $\D \in \DD(\q)$, the desired result is
$\bigcup_{i\in \D} \anchored_k(P,U,i)$.

We need to show that this result can be reconstructed from $\anchored_k(P,U,i_1)$ and $\anchored_k(P,U,i_2)$, where $i_1=\min \D,\ i_2=\max \D$.

Let $(I,J,L)=([p_l\dd p_r],[p'_l\dd p'_r],[x_l\dd x_r]) \in \R_k(P,U,i)$ for some $i \in \CritPos(\D)$, $i \ne i_1,i_2$. Let us denote $x'_l=x_l+m-1$ and $x'_r=x_r+m-1$. 

Now we consider several cases on $p_l,x_l,p'_r,x'_r$ as listed in \cref{fct:max}.

\proofsubparagraph{Case 1.} If $(I,J,L)$ is left-$U$-locked, then by \cref{lem:after_locked}, the substrings $V[x_l\dd x_l+(k+5)q)$, $V(x'_l-(k+5)q \dd x'_l+q]$ exist and none of them contains a position from a locked fragment (we consider a non-overlap offset). By \cref{lem: non-overlap pointwise2}\eqref{itB2}, $[p_l\dd p_l+q)\cap[p_l\dd p_r] \subseteq \anchored_k(P,U,i+q)$. By \cref{lem:left_right}, $[p_l+q \dd p_r] \subseteq \anchored_k(P,U,i+q)$. Thus $[p_l\dd p_r] \subseteq \anchored_k(P,U,i+q)$. By induction on $i$, $[p_l\dd p_r] \subseteq \anchored_k(P,U,i_2)$.
(We can use induction as after each step $i \rightarrow i+q$, we are obviously still in Case 1.)

\proofsubparagraph{Case 2.}  Similarly, if $(I,J,L)$ is right-$U$-locked, then by \cref{lem:after_locked}, the substrings $V[x_r-q \dd x_r+(k+5)q)$, $V(x'_r-(k+5)q \dd x'_r]$ exist and none of them contains a position from a locked fragment. By \cref{lem: non-overlap pointwise2}\eqref{itB2}, 
$(p_r-q\dd p_r] \cap[p_l\dd p_r]\subseteq \anchored_k(P,U,i-q)$.
By \cref{lem:left_right}, $[p_l \dd p_r-q] \subseteq \anchored_k(P,U,i-q)$.
Thus $[p_l\dd p_r] \subseteq \anchored_k(P,U,i-q)$. By induction on $i$ (decreasingly), $[p_l\dd p_r] \subseteq \anchored_k(P,U,i_1)$.

\proofsubparagraph{Case 3.} By \cref{fct:max} it is enough to consider now the case that $(I,J,L)$ is simultaneously left-$V$-locked and right-$V$-locked. Similarly to the above, (by \cref{lem:after_locked}) none of the four substrings 
\begin{align*}
    U(p_l-(k+5)q \dd p_l+q],\,&\, U(p'_l-(k+5)q \dd p'_l+q],\\
    U[p_r-q \dd p_r+(k+5)q),\,&\, U[p'_r-q \dd p'_r+(k+5)q)
\end{align*} 
contains a position from a locked fragment.
By \cref{lem: non-overlap pointwise2}\eqref{itA2}, we have
\begin{align*}
 \Triples((p_r+q-\delta\dd p_r+q], (p_r'+q-\delta\dd p_r'+q], (x_r-\delta\dd x_r]) &\subseteq \anchored'_k(P,U,i+q),\\
 \Triples([p_l-q\dd p_l-q+\delta),[p_l'-q\dd p_l'-q+\delta),[x_l\dd x_l+\delta)) &\subseteq \anchored'_k(P,U,i-q).
\end{align*}
where $\delta=\min(q,p_r-p_l+1)$.
Together with \cref{lem:left_right}, we obtain
\begin{align*}
\hspace*{0.3mm} \Triples([p_l+q\dd p_r+q],[p'_l+q\dd p'_r+q],[x_l\dd x_r]) &\subseteq \anchored'_k(P,U,i+q)\text{ and}\\
\Triples([p_l-q\dd p_r-q],[p'_l-q\dd p'_r-q],[x_l\dd x_r]) &\subseteq \anchored'_k(P,U,i-q).
\end{align*}

\noindent   By induction on $i$ (decreasing),
  \[\Triples([p_l+i_1-i\dd p_r+i_1-i],[p'_l+i_1-i\dd p'_r+i_1-i],[x_l\dd x_r]) \subseteq \anchored_k'(P,U,i_1).\]
  Hence, $\anyanchored'_k(P,U,i_1) \ne \mathit{none}$.
\end{proof}

\begin{proof}[Proof of \cref{thm:main_pillar}, decision version]
If $\lambda_k \le \frac{m}{2}$, \cref{lem:correctness} and \cref{corr:ov_dec} cover the decision version of \EPSM for $\q$-non-overlap offsets and $\q$-overlap offsets, respectively. Together with \cref{lem:usemesomewhere} used for the corner case that $\lambda_k > \frac{m}{2}$, they yield a solution to a decision version of \EPSM. The decision version from \cref{thm:main_pillar} is obtained through the reduction to \EPSM of \cref{lem:red}, as the time complexities of all the algorithms in the \pillar model are $\Oh(k^5 \log^3 k)$.
\end{proof}

\subsection{Reporting Version}\label{sec:algo2rep}
\cref{algo2} is a reporting version of \cref{algo2dec}.
\cref{algo2} outputs all $\q$-non-overlap app-matches as a collection of $\Oh(k^4)$ interval chains (some of which can be single intervals).

\begin{algorithm}[h!]

\vspace*{0.2cm}
\ForEach{interval of offsets $\D \in \DD(\q)$}
{

\smallskip
$i_1 := \min\,  \CritPos(\D)$;\ \  
$i_2 := \max\,  \CritPos(\D)$\;

\smallskip 
$\ZZ_1 :=  \anchored_k(P,U,i_1)$;

\smallskip $\ZZ_2 :=  \anchored_k(P,U,i_2)$\;

\smallskip
{\bf report} $\ZZ_1\cup \ZZ_2$\;

\smallskip
  \ForEach{interval $I=[p_l\dd p_r]$ \KwSty{in} $\ZZ_1$, with $p_l>0$ and $p_r+m+k\le |U|$\smallskip}{
    \If{$(\{p_l-1\}\cup I) \cap \locked(U)=\emptyset$\smallskip}{
        {\bf report} $\Chain(I,\,(i_2-i_1)/q,\,q)$\;
    }
  }

}
\caption{\bf Non-overlap case: reporting version}\label{algo2}
\vspace*{0.2cm}
\end{algorithm}

\begin{lemma}\label{lem:correctness2}
Assume that $\lambda_k \le \frac{m}{2}$.
\cref{algo2} works in $\Oh(k^4 \log \log k)$ time in the \pillar model and returns $\Oh(k^4)$ interval chains. For each $\q$-non-overlap app-match $(p,x)$, position $p$ is reported in one of the chains; moreover, only starting positions of circular $k$-edit occurrences of $P$ in $U$ are reported.
\end{lemma}

\begin{proof}
  The complexity of the algorithm is the same as of \cref{algo2dec} except for checking the condition in the if-statement. The condition can be checked offline for all intervals $[p_l \dd p_r]$ in $\ZZ_1$ at once. It suffices to sort the endpoints of locked fragments in $U$ ($\Oh(k)$ integers) together with positions $p_l-1$ and $p_r+m-k$ from all query intervals ($\Oh(k^2)$ integers). The sorting can be done in $\Oh(k^2 \log \log k)$ time~\cite{DBLP:journals/jal/Han04}. Afterwards, we can compute the predecessor and successor of each position $p_l-1$ and $p_r+m-k$ using a simple line sweep, in $\Oh(k^2)$ total time. Over all intervals of offsets $\D$ in $\DD(\q)$, this gives $\Oh(k^4 \log \log k)$ time.

  The output of the algorithm consists of $\Oh(k^4)$ intervals from the sets $\ZZ_1$ and $\ZZ_2$ and $\Oh(k^4)$ interval chains.

  The correctness proof is a continuation of the proof of \cref{lem:correctness}. 
  Cases 1 and 2 stay the same. In Case 3, we have shown that
  \[\Triples([p_l+i_1-i\dd p_r+i_1-i],[p'_l+i_1-i\dd p'_r+i_1-i],[x_l\dd x_r]) \subseteq \anchored_k'(P,U,i_1).\]
  Hence, $[p_l\dd p_r] \subseteq \Chain(J, (i_2-i_1)/q, q)$, where $J=[p_l+i_1-i\dd p_r+i_1-i]$.

    \medskip
    This shows that it is enough to report all positions of $\Chain(J, (i_2-i_1)/q, q)$ to report the interval $[p_l\dd p_r]$.
    Additionally, by induction on $i$ (increasing), we know that all of the positions of $\Chain(J, (i_2-i_1)/q, q)$ are valid solutions. It remains to show that all those positions will be returned by \cref{algo2} even if $J$ is not actually an interval returned by $\anchored_k(P,U,i_1)$.

    Notice that 
    \begin{align*}\Chain(J,(i_2-i_1)/q,q)&=J\cup \Chain([j_l\dd j_r],(i_2-i_1)/q,q)\\&=J\cup \bigcup_{p\in[j_l\dd j_r]}\Chain(\{p\},(i_2-i_1)/q,q),
    \end{align*}
    where $j_r=\max(J)$ and $j_l=\max(\min(J),j_r-q)$.
    
    Since $J$ is reported
    for $i_1$, it is enough to focus on the chain part, for which we know that $[j_l\dd j_r]$ does not intersect any locked fragment in $U$ (since we are in the right-$V$-locked case).

    Take any position $p\in[j_l\dd j_r]$. Since $(p,p',x)\in\anchored'_k(P,U,i_1)$, there must exist an element $(\bar{I},\bar{J},\bar{L})=([\bar{p}_l\dd \bar{p}_r],[\bar{p}'_l\dd \bar{p}'_r],[\bar{x}_l\dd \bar{x}_r])\in\R_k(P,U,i_1)$ such that $(p,p',x)\in \Triples(\bar{I},\bar{J},\bar{L})$.
    Now similarly to the proof of \cref{lem:correctness2} we consider three cases on $(\bar{I},\bar{J},\bar{L})$.

\proofsubparagraph{Case 1'.} If $(\bar{I},\bar{J},\bar{L})$ is both left-$V$-locked and right-$V$-locked, then \cref{algo2} will produce $\Chain(\bar{I},(i_2-i_1)/q,q)$, and hence $\Chain(\{p\},(i_2-i_1)/q,q)$ will be reported.
    
\proofsubparagraph{Case 2'.} If $(\bar{I},\bar{J},\bar{L})$ is right-$U$-locked, then we know that $\bar{p}_r\ge p+(i_2-i_1)$, as otherwise by induction on $i$, for some $i\in\CritPos(D)$ we would have that $(p+(i-i_1),p'+(i-i_1),x)\in \anchored'_k(P,U,i)$, such that $p+(i-i_1)$ is at most $q$ positions away from a locked fragment but $x$ is at most $q$ positions away from a locked fragment in $V$. This would contradict $i-j_1$ being a $(2q+k)$-non-overlap offset.
Hence, $\Chain(\{p\},(i_2-i_1)/q,q)\subseteq \bar{I}$, and thus is reported.

\proofsubparagraph{Case 3'.} If $(\bar{I},\bar{J},\bar{L})$ is right-$V$-locked and left-$U$-locked, then similarly to Case 1 in \cref{lem:correctness}, we can show that
\[\Triples([\bar{p}_l\dd \bar{p}_r+q],[\bar{p}'_l\dd \bar{p}'_r+q],[\bar{x}_l-q\dd \bar{x}_r])\subseteq \anchored'_k(P,U,i_1+q).\] By induction on $i$,
\[\Triples([\bar{p}_l\dd \bar{p}_r+(i_2-i_1)],[\bar{p}'_l\dd \bar{p}'_r +(i_2-i_1)],[\bar{x}_l-(i_2-i_1)\dd\bar{x}_r])\subseteq\anchored'_k(P,U,i_2),\]
and this set contains $\Chain(\{p\},(i_2-i_1)/q,q)$.
\end{proof}

\begin{proof}[Proof of \cref{thm:main_pillar}, decision version]
The reporting version of \cref{thm:main_pillar} follows from the reporting version of the overlap case (\cref{lem:case1}), the correctness and the complexity of \cref{algo2} (\cref{lem:correctness2}),
the usage of \cref{lem:usemesomewhere} for the corner case when $\lambda_k>\frac{m}{2}$, and the reduction to \EPSM (\cref{lem:red}).
\end{proof}

\begin{remark}
In both versions (decision, reporting), the bottleneck of the algorithm's running time is the overlap case, while the most technically demanding part is the non-overlap case.
\end{remark}

\section{\texorpdfstring{$k$-Edit}{k-Edit} CPM in Other Settings}\label{sec:results}
\cref{thm:main_pillar} is stated in the \pillar model. In the standard setting, all \pillar operations can be implemented in $\Oh(1)$ time after $\Oh(n)$ preprocessing \cite[Section 3]{DBLP:journals/corr/abs-2208-08915}; this yields \cref{thm:main}.

We now present our results for the internal, dynamic, fully compressed, and quantum settings.
In each case, in the reporting version of the problem, the output is represented as a union of $\Oh((|T|/|P|) \cdot k^6)$ interval chains.

With the same implementations of operations in the internal setting as in the standard setting, we obtain an efficient implementation. 

\begin{theorem}[Internal Setting]
Given two substrings $P$ and $T$ of a length-$n$ string $S$, reporting and decision versions of $k$-Edit CPM for $P$ and $T$ can be solved in $\Oh((|T|/|P|)k^6)$ time and $\Oh((|T|/|P|)k^5 \log^3 k)$ time, respectively, after $\Oh(n)$ preprocessing on $S$.
\end{theorem}

Let $\mathcal{X}$ be a growing collection of non-empty persistent strings; it is initially empty, and then undergoes updates by means of the following operations:
\begin{itemize}
    \item $\texttt{Makestring}(U)$: Insert a non-empty string $U$ to $\mathcal{X}$
    \item $\texttt{Concat}(U,V)$: Insert string $UV$ to $\mathcal{X}$, for $U,V\in \mathcal{X}$
    \item $\texttt{Split}(U,i)$: Insert $U[0\dd i)$ and $U[i\dd |U|)$ to $\mathcal{X}$, for $U\in\mathcal{X}$ and $i\in[0\dd |U|)$.
\end{itemize}

By $N$ we denote an upper bound on the total length of all strings in $\mathcal{X}$ throughout all updates executed by an algorithm.
A collection $\mathcal{X}$ of non-empty persistent strings of total length $N$ can be dynamically maintained with operations $\mathtt{Makestring}(U)$,  $\mathtt{Concat}(U,V)$, $\mathtt{Split}(U,i)$ requiring time $\Oh(\log N+|U|)$, $\Oh(\log N)$ and $\Oh(\log N)$, respectively, so that \pillar operations can be performed in time $\Oh(\log^2 N)$.
All stated time complexities hold with probability $1-1/N^{\Omega(1)}$; see \cite{DBLP:conf/soda/GawrychowskiKKL18,FOCS20}. Moreover, Kempa and Kociumaka~\cite[Section 8 in the arXiv version]{DBLP:conf/stoc/KempaK22} presented an alternative deterministic implementation, which supports operations $\texttt{Makestring}(U)$, $\texttt{Concat}(U,V)$, $\texttt{Split}(U,i)$ in $\Oh(|U|\log^{\Oh(1)}\log N)$, $\Oh(\log|UV|\log^{\Oh(1)}\log N)$, and $\Oh(\log|U|\log^{\Oh(1)}\log N)$ time, respectively, so that \pillar operations can be performed in time $\Oh(\log N \log^{\Oh(1)}\log N)$. With these implementations, we obtain the following result.

\begin{theorem}[Dynamic Setting]\label{the:CPMdynamic}
A collection $\mathcal{X}$ of non-empty persistent strings of total length $N$ can be dynamically maintained with operations $\mathtt{Makestring}(U)$, $\mathtt{Concat}(U,V)$, $\mathtt{Split}(U,i)$ requiring time $\Oh(\log N+|U|)$, $\Oh(\log N)$ and $\Oh(\log N)$, respectively, so that, given two strings $P,T\in \mathcal{X}$ and an integer threshold $k>0$,
we can solve $k$-Edit CPM in $\Oh((|T|/|P|)\cdot k^{6}\log^2 N)$ time for the reporting variant and $\Oh((|T|/|P|)\cdot k^{5} \log^3 k\log^2 N)$ time for the decision variant.
All stated time complexities hold with probability $1-1/N^{\Omega(1)}$.
Randomization can be avoided at the cost of a $\log^{\Oh(1)} \log N$ multiplicative factor in all the update times, with $k$-Edit CPM queries answered in $\Oh((|T|/|P|)\cdot k^{6}\log N\log^{\Oh(1)}\log N)$ time (reporting version) or $\Oh((|T|/|P|)\cdot k^{5} \log^3 k\log N\log^{\Oh(1)}\log N)$ time (decision version).
\end{theorem}

\newcommand{\gen}{\mathit{gen}}

A straight line program (SLP) is a context-free grammar $G$ that consists of a set $\Sigma$ of terminals and a set $N_G = \{A_1,\dots,A_n\}$ of non-terminals such that each $A_i \in N_G$ is associated with a unique production rule
$A_i \rightarrow f_G(A_i) \in (\Sigma \cup \{A_j : j < i\})^*$. We can assume without loss of generality that each production rule is of the form $A \rightarrow BC$ for some symbols $B$ and~$C$ (that is, the given SLP is in Chomsky normal form).
Every symbol $A \in S_G:=N_G \cup\Sigma$ generates a unique string, which we denote by $\gen(A) \in \Sigma^*$. The string $\gen(A)$ can be obtained from $A$ by repeatedly replacing each non-terminal with its production.
We say that~$G$ generates $\gen(G) := \gen(A_n)$.

In the fully compressed setting, given a collection of straight-line programs (SLPs) of total size $n$ generating strings
of total length $N$, each \pillar operation can be performed in $\Oh(\log^2 N \log \log N)$ time after an $\Oh(n \log N)$-time preprocessing \cite[Section 3]{DBLP:journals/corr/abs-2208-08915}. If we applied \cref{thm:main} directly in the fully compressed setting, we would obtain $\Omega(N/M)$ time, where $N$ and $M$ are the uncompressed lengths of the text and the pattern, respectively. Instead, we can adapt an analogous procedure provided in~\cite[Section 7.2]{FOCS20}
for (non-circular) pattern matching with edits to obtain the following result.

\begin{theorem}[Fully Compressed Setting]
Let $G_T$ denote a straight-line program of size~$n$ generating a string $T$, let $G_P$ denote a
straight-line program of size $m$ generating a string~$P$, let $k > 0$ denote an integer threshold, and set $N := |T|$ and $M := |P|$.
We can solve $k$-Edit CPM in $\Oh(m \log N + n k^{6} \log^2 N \log \log N)$ time (counting version) or $\Oh(m \log N + n k^{5} \log^3 k \log^2 N \log \log N)$ time (decision version). A representation of the occurrences in the form of interval chains can be returned in $\cO((N/M)\cdot k^6)$ extra time.
\end{theorem}

We say an algorithm on an input of size $n$ succeeds \emph{with high probability} if the success probability can be made at least $1-1/n^c$ for any desired constant $c>1$.

In what follows, we assume the input strings can be accessed in a quantum query model~\cite{AMB04,DBLP:journals/tcs/BuhrmanW02}.
We are interested in the time complexity of our quantum algorithms~\cite{BBCplus}.

\begin{observation}[{\cite[Observation 2.3]{DBLP:conf/soda/JinN23}}]
For any two strings $S,T$ of length at most $n$,
$\LCP(S, T)$ or $\LCP_R(S, T)$ can be computed in $\ctO(\sqrt{n})$ time in the quantum model with high probability.
\end{observation}

Hariharan and Vinay~\cite{DBLP:journals/jda/HariharanV03} gave a near-optimal quantum algorithm for the decision version of exact PM. We formalize this next.

\begin{theorem}[\cite{DBLP:journals/jda/HariharanV03}]\label{the:quantumPM}
The decision version of PM can be solved in $\ctO(\sqrt{n})$ time in the quantum model with high probability.
If the answer is YES, then the algorithm returns a witness occurrence.
\end{theorem}

By employing~\cref{the:quantumPM} and binary search to find the period of $S$~\cite{DBLP:conf/soda/KociumakaRRW15} and thus its full list of occurrences expressed as an arithmetic progression in $T$, we obtain the following.

\begin{observation}
For any two strings $S,T$ of length at most $n$, with $|T| \le 2|S|$,
$\IPM(S, T)$ can be computed in $\ctO(\sqrt{n})$ time in the quantum model with high probability.
\end{observation}

All other \pillar operations are performed trivially in 
$\cO(1)$ quantum time. Thus while all \pillar operations can be implemented in $\Oh(1)$ time after $\Oh(n)$-time preprocessing in the standard setting by a classic algorithm, in the quantum setting, all \pillar operations can be implemented in $\ctO(\sqrt{m})$ quantum time \emph{with no preprocessing}, as we always deal with strings of length $\cO(m)$. We obtain the following results.

\begin{theorem}[Quantum Setting]\label{thm:quantum}
The reporting version of the $k$-Edit CPM problem can be solved in $\ctO((n/\sqrt{m})k^6)$ time in the quantum model with high probability.
The decision version of the $k$-Edit CPM problem can be solved in $\ctO((n/\sqrt{m})k^5)$ time in the quantum model with high probability.
\end{theorem}

\bibliographystyle{plainurl}
\bibliography{references}

\end{document}

%% file: _fig_setting.tex
\begin{tikzpicture}[xscale=0.2,yscale=0.44]
\fill[green!20!white,xshift=-1cm] (12.5,0) rectangle (35.5,1);
\draw (12,0) node[below] {$p$};
\draw[blue,thick,latex-latex] (12.5,-0.4) -- (16.5,-0.4);
\foreach \x/\c in {0/a,1/b,2/c,3/d,4/e,5/f,6/g}{
  \draw (\x,-0.1) node[above] {\small \vphantom{\texttt{g}}{\texttt{\c}}};
}
\foreach \x/\c in {7/y}{
  \draw (\x,-0.1) node[above] {\small \vphantom{\texttt{g}}{\textcolor{red}{\texttt{\c}}}};
}
\foreach \x/\c in {7/h, 8/a,9/b, 10/c,11/d,12/e,13/f,14/g,15/h, 16/a,17/b,18/c}{
  \draw[xshift=1cm] (\x,-0.1) node[above] {\small \vphantom{\texttt{g}}{\texttt{\c}}};
}

\draw[red] (19.5,0) -- (19.5,1);
\foreach \x/\c in {20/e,21/f,22/g,23/h, 24/a,25/b,26/c,27/d,28/e,29/f,30/g,31/h,32/a,33/b,34/c,35/d}{
  \draw[xshift=0cm] (\x,-0.1) node[above] {\small \vphantom{\texttt{g}}{\texttt{\c}}};
}
\draw (-0.5,0) rectangle (35.5,1);
\draw (-0.5,0.5) node[left] {$U$};
\begin{scope}
\clip (-0.5,1) rectangle (35.5,3);
\draw[xshift=-0.5cm] (0,1) .. controls (3,2) and (6,2) .. (9,1);
\draw[xshift=8.5cm] (0,1) .. controls (3,2) and (5,2) .. node[above] {$Q$} (8,1);
\draw[xshift=16.5cm] (0,1) .. controls (2,2) and (5,2) .. (7,1);
\draw[xshift=23.5cm] (0,1) .. controls (3,2) and (5,2) .. (8,1);
\draw[xshift=31.5cm] (0,1) .. controls (3,2) and (5,2) .. (8,1);
\end{scope}

\begin{scope}[yshift=-5cm,xshift=1cm]
\draw (11,0) node[below] {$x$};
\draw[blue,thick,latex-latex] (11.5,-0.4) -- (16.5,-0.4);
\fill[green!20!white,xshift=-1cm] (11.5,0) rectangle (36.5,1);
\foreach \x/\c in {5/f,6/g,7/h, 8/a,9/b,10/c,11/d,12/e,13/f
}{
  \draw (\x,-0.1) node[above] {\small \vphantom{\texttt{g}}{\texttt{\c}}};
}
\foreach \x/\c in {14/z,
}{
  \draw (\x,-0.1) node[above] {\small \vphantom{\texttt{g}}{\textcolor{red}{\texttt{\c}}}};
}
\foreach \x/\c in {14/g,15/h, 16/a,17/b,18/c,19/d,20/e,21/f,22/g,23/h, 24/a,25/b,26/c,27/d,28/e,29/f,30/g,31/h,
  32/a,33/b,34/c,35/d,36/e
}{
  \draw[xshift=1cm] (\x,-0.1) node[above] {\small \vphantom{\texttt{g}}{\texttt{\c}}};
}
\draw (4.5,0) rectangle (37.5,1);
\draw (4.5,0.5) node[left] {$V$};
\begin{scope}
\clip (4.5,1) rectangle (37.5,3);
\draw[xshift=-0.5cm] (0,1) .. controls (3,2) and (5,2) .. (8,1);
\draw[xshift=7.5cm] (0,1) .. controls (3,2) and (6,2) .. (9,1);
\draw[xshift=16.5cm] (0,1) .. controls (3,2) and (5,2) .. (8,1);
\draw[xshift=24.5cm] (0,1) .. controls (3,2) and (5,2) .. (8,1);
\draw[xshift=32.5cm] (0,1) .. controls (3,2) and (5,2) .. (8,1);
\end{scope}
\begin{scope}[yshift=-1cm]
\draw[densely dashed] (15.5,-1) rectangle (40.5,0);
\draw (27.5,-0.5) node {$P$};
\draw[xshift=-24cm,densely dashed] (39.5,-1) -- (14.5,-1) -- (14.5,0) -- (39.5,0);
\draw[xshift=-27cm] (27.5,-0.5) node {$P$};
\draw (5,-1) node[blue!50!black,below] {$\alpha$};
\draw[densely dotted] (4.5,-1) -- (4.5,1);
\draw (37,-1) node[blue!50!black,below] {$\beta$};
\draw[densely dotted] (37.5,-1) -- (37.5,1);
\end{scope}

\draw[xshift=-9cm,latex-latex] (7.5,2) -- node[above] {$r$} (12.5,2);
\draw[latex-latex] (7.5,2) -- node[above] {$r$} (12.5,2);
\draw (4.5,2.8) -- (4.5,3) -- node[above] {$Q'$} (12.5,3) -- (12.5,2.8);
\draw[densely dotted] (4.5,2.8) -- (4.5,1)  (12.5,2.8) -- (12.5,1)  (7.5,2) -- (7.5,1);

\end{scope}

\end{tikzpicture}

%% file: _fig_uv.tex
\begin{tikzpicture}[xscale=0.22,yscale=0.48]

\begin{scope}
\foreach \c [count=\i] in {d, e, f, g, h, a, b, c, d, e, f, g, h, a, b, c, d, e, f, g, h, a, b, c, d, e, f, g, h, a, b, c, d, e, f, g, h, a, b, c, d, e, f, g, h}{
  \draw (\i,-0.1) node[above] (c\i) {\small \vphantom{\texttt{g}}{\texttt{\c}}};
}
\foreach \x in {4, 12, 20} {
  \draw[blue!60!green,line width=0.1cm] ($(c\x)+(-0.5cm,-0.7cm)$) -- +(6cm,0);
}
\begin{scope}
\clip ($(c1)+(-0.6, 0.4)$) rectangle ($(c45)+(0.5, 1.4)$);
\foreach \i in {0,...,5} {
  \draw[xshift=(\i*8cm)-0.5cm,yshift=1cm] (1,0) .. controls (3,1) and (7,1) .. (9,0);
}
\end{scope}
\node at (12.5cm,2.4cm) {$Q$};
\draw ($(c1)+(-0.5, -0.5)$) rectangle ($(c45)+(0.5, 0.5)$);
\node at ($(c1)+(-2,0)$) {$U$};
\end{scope}

\draw[|<->|,yshift=-1.5cm] (0.5cm,0)--+(5cm,0) node[midway,above] {$r$};

\draw[latex-latex] (22cm,-1cm) node[right] {$x+r$}--(22cm,-2.5cm) node[right] {$x$};

\begin{scope}[yshift=-4cm]
\foreach \c [count=\i] in {a, b, c, d, e, f, g, h, a, b, c, d, e, f, g, h, a, b, c, d, e, f, g, h,}{
  \draw (\i+5,-0.1) node[above] (d\i) {\small \vphantom{\texttt{g}}{\texttt{\c}}};
}
\draw ($(d1)+(-0.5, -0.5)$) rectangle ($(d24)+(0.5, 0.5)$);
\draw[xshift=(5cm+3cm-0.5cm,yshift=1cm] (1,0) .. controls (3,1) and (7,1) .. (9,0) node[midway,above] {$Q$};

\draw[xshift=(5cm-0.5cm,yshift=-0.13cm] (1,0) .. controls (3,-1) and (7,-1) .. (9,0) node[midway,below] {$Q'$};

\node at ($(d1)+(-2,0)$) {$V$};

\end{scope}

\end{tikzpicture}

%% file: _fig_july6.tex
\begin{tikzpicture}
\def\lenH{0.2}
\def\lenD{1.0}
\tikzstyle{white}=[fill=white]
\tikzstyle{gray}=[fill=white!80!black]

\pgfmathsetmacro{\prevX}{0}
\foreach \width/\style [count=\i] in {
  0.5/gray, 1/gray,
  1/white, 1/white, 1/white,
  1.2/gray,0.9/gray,
  1/white, 1/white, 1/white,
  0.7/gray,1/gray,
  1/white, 0.7/gray%
  } {
  \pgfmathsetmacro{\currX}{\prevX + \width};
  \coordinate (l\i) at (\prevX, 0);
  \coordinate (ll\i) at (\prevX, \lenH);
  \coordinate (r\i) at (\currX, \lenH);
  \coordinate (d\i) at (\currX, 0);
  \draw[\style] (l\i) rectangle (r\i);
  \global\let\prevX=\currX
}

\node[below,yshift=-0.15cm] at ($(l1)!0.5!(r2)$) {$L_1$};
\node[below,yshift=-0.15cm] at ($(l6)!0.5!(r7)$) {$L_2$};
\node[below,yshift=-0.15cm] at ($(l11)!0.5!(r12)$) {$L_3$};
\node[below,yshift=-0.15cm] at ($(l14)!0.5!(r14)$) {$L_4$};

\foreach \i in {3,4,5, 8,9,10, 13} {
  \coordinate (uL\i) at ($(l\i)+(0,\lenD)$);
  \coordinate (uR\i) at ($(d\i)+(0,\lenD)$);
  \coordinate (m1) at ($(uL\i)!0.3!(uR\i)$);
  \coordinate (m2) at ($(uL\i)!0.7!(uR\i)$);
  \draw (uL\i) .. controls ($(m1)+(0,0.2cm)$) and ($(m2)+(0,0.2cm)$) .. (uR\i)
    node[midway,above] {$Q$};
  \draw[dotted] (uL\i)--(ll\i);
  \draw[dotted] (uR\i)--(r\i);
}

\foreach \i/\dA/\dB in {2/0.1/0, 6/0/-0.1, 7/-0.1/0, 11/0/0.2, 12/0.2/0} {
  \coordinate (uL\i) at ($(l\i)+(\dA,\lenD)$);
  \coordinate (uR\i) at ($(d\i)+(\dB,\lenD)$);
  \coordinate (m1) at ($(uL\i)!0.3!(uR\i)$);
  \coordinate (m2) at ($(uL\i)!0.7!(uR\i)$);
  \draw (uL\i) .. controls ($(m1)+(0,0.2cm)$) and ($(m2)+(0,0.2cm)$) .. (uR\i)
    node[midway,above] {$Q$};
  \draw[dotted] (uL\i)--(ll\i);
  \draw[dotted] (uR\i)--(r\i);
}

\coordinate (uL1) at ($(l1)+(0.05,\lenD)+(0,0.1)$);
\coordinate (uR1) at ($(d1)+(0.1,\lenD)$);
\coordinate (m1) at ($(uL1)!0.5!(uR1)$);
\draw (uL1) .. controls ($(m1)+(0,0.1cm)$) .. (uR1);
\draw[dotted] (uL1)--(ll1);

\coordinate (uL14) at ($(l14)+(0,\lenD)$);
\coordinate (uR14) at ($(d14)+(-0.1,\lenD)+(0,0.1)$);
\coordinate (m1) at ($(uL14)!0.6!(uR14)$);
\draw (uL14) .. controls ($(m1)+(0,0.1cm)$) .. (uR14);
\draw[dotted] (uR14)--(r14);

\end{tikzpicture}

%% file: _fig_anchor.tex
\begin{tikzpicture}[yscale=0.35,xscale=0.65]
\draw (0,0) rectangle (10,1);
\draw (5,0) -- (5,1);
\draw (2.5,2) node[above] {\small $P$};
\draw (7.5,2) node[above] {\small $P$};
\draw (8,0) -- (8,1);
\draw (6.5,0.5) node {\small $P_1$};
\draw (8.5,0.5) node {\small $P_2$};

\draw[blue] (3,-2) rectangle (8,-1);
\fill[white!70!brown] (3,-2) rectangle (8,-1);
\draw (1.5,-2) rectangle (9,-1);
\draw (1.5,-1.5) node[left] {$V$};
\draw[blue] (3,-2) rectangle (8,-1);
\draw (4,-1.5) node {\small $P_2$};
\draw (6.5,-1.5) node {\small  $P_1$};
\draw[blue] (5,-2) -- (5,-1);
\draw (3,-2) node[below left=-0.1cm] {$x$};

\fill[white!70!brown] (2.8,-5) rectangle (7.6,-4);
\draw (0.5,-5) rectangle (8.5,-4);
\draw (0.5,-4.5) node[left] {$U$};
\draw (3.75,-4.5) node {\small $X_2$};
\draw[blue] (2.8,-5) rectangle (7.6,-4);
\draw (2.8,-5) node[below] {$p$};
\draw[densely dotted,blue] (2.8,-4) -- (3,-2)  (7.6,-4) -- (8,-2);
\draw[blue] (5,-5) -- (5,-4);
\draw[-latex] (5,-6) -- (5,-5.1);
\draw (5,-6) node[below] {anchor $m-\alpha+\Delta+\delta$};

\draw[-latex] (0.5,-3.2) -- node[above] {\small $\Delta$} (1.3,-3.2);
\draw[densely dotted] (0.5,-4) -- (0.5,-3)  (1.3,-3) -- (1.5,-2);

\draw[densely dotted] (1.5,-1) -- (1.5,1);
\draw (1.5,1.2) node[above] {$\alpha$};
\draw[densely dotted] (9,-1) -- (9,1);
\draw (9,1) node[above] {$\beta$};

\draw[snake=brace] (1.5,1.1) -- node[above] {$m-\alpha$} (5,1.1);
\draw[densely dotted] (5,1) -- (5,3);

\end{tikzpicture}

%% file: _fig_algorithm.tex
\begin{tikzpicture}[scale=0.9]

    \def\lenH{0.2}
    \tikzstyle{dot}=[inner sep=0.06cm, circle, draw, fill=black]
    \tikzstyle{dot2}=[inner sep=0.06cm, circle, draw, fill=green]
    
    \pgfmathsetmacro{\prevX}{0}
    \foreach \w [count=\i] in {0.3, 1, 1, 1, 1, 1, 1, 1, 1, 1, 1, 1, 1} {
      \pgfmathsetmacro{\currX}{\prevX + \w};
      \coordinate (l\i) at (\prevX,0);
      \coordinate (ll\i) at (\prevX,\lenH);
      \coordinate (r\i) at (\currX,0);
      \coordinate (rr\i) at (\currX,\lenH);
      \draw (\prevX, 0) rectangle (\currX, \lenH);
      \global\let\prevX=\currX
    }
    \node[left] at ($(l1)!0.5!(ll1)$) {$U$};
    \node[dot2] at ($(l3)+(0,0)$) (i1) {};
    \node[dot2] at ($(l4)+(0,0)$) {};
    \node[dot2] at ($(l5)+(0,0)$) {};
    \node[dot2] at ($(l6)+(0,0)$) {};
    \node[dot2] at ($(l7)+(0,0)$) {};
    \node[dot2] at ($(l8)+(0,0)$) {};
    \node[below right=-0.1cm] at (i1) {\small $i_1$};
    \node[dot2] at ($(l9)+(0,0)$) (i2) {};
    \node[below] at (i2) {\small $i_2$};
    \node[above=0.1cm] at ($(ll6)!0.5!(rr6)$) {$Q$};

    \node[dot] at ($(l2)+(0.2,0)$) (jj1) {};
    \node[dot] at ($(l12)+(0.7,0)$) (jj2) {};
    \node[above=0.4cm] at (jj1) {\small $j_1+d_1-k$};
    \node[above=0.4cm] at (jj2) {\small $j_2+d_2-k$};

    \foreach \i/\c in {2/black,3/red,4/red,5/red,6/black,7/black,8/black,9/red,10/red,11/red,12/black,13/black} {
      \coordinate (m1) at ($(l\i)!0.3!(r\i)$);
      \coordinate (m2) at ($(l\i)!0.7!(r\i)$);
      \draw[\c] ($(l\i)+(0,\lenH)$) .. controls ($(m1)+(0,2*\lenH)$) and ($(m2)+(0,2*\lenH)$) .. ($(r\i)+(0,\lenH)$);
    }

    \begin{scope}[xshift=1cm,yshift=-1.5cm]
    \node[left] at (0,0.5*\lenH) {$V$};
    \draw (0, 0) rectangle (3, \lenH);
    \draw (3, 0) rectangle (4, \lenH);
    \draw (4, 0) rectangle (5, \lenH);
    \draw (5, 0) rectangle (6, \lenH);
    \draw (6, 0) rectangle (9, \lenH);
    \coordinate (s1) at (3,0);
    \coordinate (ss1) at (3,\lenH);
    \coordinate (s2) at (6,0);
    \coordinate (ss2) at (6,\lenH);
    \node[dot] at (s1) {};
    \node[below] at (s1) {\small $j_1$};
    \node[dot] at (s2) {};
    \node[below] at (s2) {\small $j_2$};
    \foreach \i in {0,...,2} {
      \coordinate (m0) at ($(3+\i,\lenH)$);
      \coordinate (m1) at ($(3+\i+0.3,2*\lenH)$);
      \coordinate (m2) at ($(3+\i+0.7,2*\lenH)$);
      \draw[red] (m0) .. controls (m1) and (m2) .. +(1,0);
    }
    \node [above] at (4.5,2*\lenH) {\small sample $V[j_1\dd j_2]$};
    \end{scope}

    \draw[-latex] (ss1)--(jj1);
    \draw[-latex] (ss2)--(jj2);

\end{tikzpicture}

%% file: _fig_locked_eq.tex
\begin{tikzpicture}

\coordinate (left) at (0, 0);
\coordinate (right) at (11, 0.3);

\coordinate (a1) at (2, 0);
\coordinate (b1) at (3, 0.3);
\coordinate (a2) at (4.5, 0);
\coordinate (b2) at (5.5, 0.3);
\coordinate (a3) at (9, 0);
\coordinate (b3) at (10, 0.3);
\coordinate (d) at (0.8, 0);
\coordinate (dd) at (7.0, 0);

\draw (left) rectangle (right);
\foreach \i in {1,...,3} {
  \draw[draw=black,fill=blue!30!white] (a\i) rectangle (b\i);
}

\draw[|<->|] ($(d)+(0,-0.3cm)$)--($(a1)+(0,-0.3cm)$) node[midway,below] {$d$};
\draw[|<->|] ($(dd)+(0,-0.3cm)$)--($(b2 |- a2)+(0,-0.3cm)$) node[midway,below] {$d'$};

\coordinate (ddR) at ($(dd)+(0,+0.6cm)$);
\coordinate (ddL) at ($(ddR)+(-6.7cm,0)$);

\coordinate (dL) at ($(d)+(0,+1.2cm)$);
\coordinate (dR) at ($(dL)+(6.7cm,0)$);

\draw[|<->|] (ddL)--(ddR) node[midway,above] {$U''$};
\draw[|<->|] (dL)--(dR) node[midway,above] {$U'$};

\draw[dotted] ($(d)+(0,-0.3cm)$)--(dL);
\draw[dotted] ($(dd)+(0,-0.3cm)$)--(ddR);

\end{tikzpicture}